\documentclass[journal,twoside,web]{ieeecolor}
\usepackage{generic}
\usepackage{cite}
\usepackage{balance}
\usepackage{amsmath,amssymb,amsfonts}
\usepackage{graphicx}
\usepackage{blox}

\newtheorem{secthm}{Theorem}[section]

\newtheorem{secex}[secthm]{Example}
\newtheorem{secprop}[secthm]{Proposition}

\newtheorem{secdefn}[secthm]{Definition}
\newtheorem{secrem}[secthm]{Remark}
\newtheorem{secasm}[secthm]{Assumption}
\newtheorem{secsasm}[secthm]{Standing Assumption}

\usepackage{bbm}
\usepackage{amsfonts}
\usepackage{booktabs}
\usepackage{dsfont}
\usepackage{siunitx}
\usepackage{tikz}
\usetikzlibrary{arrows,automata}
\usetikzlibrary{arrows,shapes,backgrounds,calc,positioning,patterns}
\usetikzlibrary{calc,arrows,shapes,backgrounds,calc,positioning,patterns,decorations.pathmorphing,decorations.markings,mindmap,trees}
\tikzstyle{block} = [draw, rectangle, minimum height=2em, minimum
width=4em] \tikzstyle{sum} = [draw, fill=blue!20, circle, node
distance=1cm] \tikzstyle{input} = [coordinate] \tikzstyle{output} =
[coordinate] \tikzstyle{pinstyle} = [pin edge={to-,thin,black}]
\usepackage[american,cute inductors,smartlabels]{circuitikz}
\usetikzlibrary{arrows,automata}
\usepackage[american,cuteinductors,smartlabels]{circuitikz}
\usetikzlibrary{calc}
\ctikzset{bipoles/thickness=1} \ctikzset{bipoles/length=0.8cm}
\ctikzset{bipoles/diode/height=.375}
\ctikzset{bipoles/diode/width=.3}
\ctikzset{tripoles/thyristor/height=.8}
\ctikzset{tripoles/thyristor/width=1}
\ctikzset{bipoles/vsourceam/height/.initial=.7}
\ctikzset{bipoles/vsourceam/width/.initial=.7} \tikzstyle{every
node}=[font=\small] \tikzstyle{every path}=[line width=0.8pt,line
cap=round,line join=round]

\newcommand{\bR} { {\mathbb R}}
\newcommand{\bZ} { {\mathbb Z}}

\newcommand{\cE} { {\mathcal E}}
\newcommand{\cH} { {\mathcal H}}
\newcommand{\cJ} { {\mathcal J}}

\newcommand{\cR} { {\mathcal R}}

\newcommand{\1}{\mbox{1}\hspace{-0.25em}\mbox{l}}

\def\red{\hfill $\lhd$}

\usepackage{xcolor}

\allowdisplaybreaks[4]

\begin{document}
\title{Krasovskii Passivity for Sampled-data Stabilization and Output Consensus}
\author{Yu Kawano,  Alessio Moreschini, Michele Cucuzzella
\thanks{The work of Yu Kawano was partially supported by JSPS KAKENHI Grant Number JP21K14185. The work of Alessio Moreschini has been supported by the EPSRC grant “Model Reduction from Data”, Grant No. EP/W005557. The work of Michele Cucuzzella is part of the project NODES which has received funding from the MUR--M4C2 1.5 of PNRR with grant agreement no. ECS00000036.}
\thanks{Yu Kawano is with the Graduate School of Advanced Science and Engineering, Hiroshima University, Higashi-hiroshima 739-8527, Japan (e-mail: ykawano@hiroshima-u.ac.jp).}
\thanks{Alessio Moreschini is with the Department of Electrical and Electronic Engineering, Imperial College London, SW7 2AZ London, U.K. (e-mail: a.moreschini@imperial.ac.uk).}
\thanks{Michele Cucuzzella is with the Department of Electrical, Computer and Biomedical Engineering, University of Pavia, 27100 Pavia, Italy (e-mail: michele.cucuzzella@unipv.it)}
}

\maketitle

\begin{abstract}
In this paper, we establish the novel concept of Krasovskii passivity for sampled discrete-time nonlinear systems, enabling Krasovskii-passivity-based control design under sampling. We consider two separate control objectives: stabilization and output consensus, where the latter is studied under the presence  of an unknown constant disturbance. Inspired by methodologies in the continuous-time case, we develop sampled-data control schemes for each control objective based on Krasovskii passivity. The proposed sampled discrete-time controllers are respectively  validated through simulations on a DC microgrid of boost converters and a DC microgrid of buck converters whose continuous-time models and their implicit midpoint discretizations are Krasovskii passive in each time scale. 
\end{abstract}


\section{Introduction}
\IEEEPARstart{P}{assivity} provides an input-output approach to analysis and control of nonlinear dynamical systems. As leveraged by books~ \cite{Vidyasagar:81,Schaft:17,HC:11,SJK:12}, passivity is nowadays a mature theory, playing a central role in systems and control. Nevertheless, in several nonlinear applications, different notions of passivity are still being investigated by looking at inputs or outputs that are different from the standard ones, \emph{e.g.}, counterclockwise~\cite{Angeli:06}, differential passivity~\cite{Schaft:13}, shifted passivity~\cite{JOG:07,HAP:11,Simpson:18}, and incremental passivity~\cite{PM:08}. 

Along this line of research, the concept of Krasovskii passivity~\cite{KKS:21} (also called $\delta$-passivity~ \cite{SA:21}) has been introduced for continuous-time nonlinear systems, taking the time-derivative of the actual input to perform passivity analysis. Krasovskii passivity provides a systematic passivity framework for integral-type control design. As exemplified by the stabilization of DC-Zeta converters~\cite{KKS:21}, boost converters~\cite{CLK:19} or more general Brayton-Moser systems \cite{KCS:21}, the Krasovskii passivity framework significantly simplifies stabilizing control design. Also, Krasovskii passivity has been applied to output consensus control~\cite{KCF:22}, to acieve for example current sharing in DC microgrids under unknown power consumption. In practice, these control methodologies for continuous-time systems are implemented after temporal discretization of the considered plant. Hence, we need additional performance analysis for the sampled closed-loop system.

\smallskip

\subsubsection*{Literature review}
Temporal discretization of a passive system has been studied in the context of numerical integration. Notably, the standard discretization of continuous-time dynamics by means of zero-order hold and ideal sampler does not yield in general discrete-time dynamics satisfying the standard passivity inequality with respect to the same continuous-time output~\cite{Oishi:10,KHKSA:13,XAGZ:17,MMM:21,MMMN:22}. Accordingly, if a given passive system is expressed as a port-Hamiltonian system, the underlying geometric structure is generally destroyed~\cite{LA:06,MNM:22}. To overcome this issue, geometric and symplectic integration schemes~\cite{HLW:06} are employed with the aim of preserving the invariants of motion, and thus ensuring faithful performances of the long-time behavior while retaining passivity at the sampling instants. Also in~\cite{MMN:19,MNM:22,MQR:99}, the discrete gradient function has been employed to define port-Hamiltonian dynamics based on the exact flow of the system, while in~ \cite{LA:06,KL:20}, the implicit midpoint method or one-stage Gauss-Legendre collocation has been employed to preserve the symplectic geometric structures. However, none of the aforementioned papers has dealt with Krasovskii passivity, except for our preliminary work~\cite{KMC:22} for linear port-Hamiltonian systems (PHSs).

\smallskip

\subsubsection*{Contributions} To develop sampled-data control schemes based on Krasovskii passivity, we introduce the novel concept of Krasovskii passivity for sampled discrete-time nonlinear systems. Inspired by the continuous-time property, we define Krasovskii passivity by taking the temporal discretization of the time derivative of the actual input as the new input
involved in the passivity analysis. %This discrete-time Krasovskii passivity concept is a natural extension of the continuous-time one, since implicit midpoint discretizations of boost converters and islanded DC microgrids are again Krasovskii passive in the sampled discrete-time sense. 
To get a better understanding of the proposed Krasovskii passivity concept, we also investigate its connection with relevant passivity concepts originally defined for continuous-time systems: incremental passivity and shifted  passivity~\cite{JOG:07,PM:08,Simpson:18}. In particular, we show the following implications: incremental passivity $\implies$ Krasovskii passivity $\implies$ shifted passivity with respect to a suitable output function. The latter can be expected by a similar implication in the continuous-time case~\cite{KKS:21}. However, the first comes from discrete-time nature. In this paper, we study nonlinear systems while linear PHSs are used for illustrating our results.  %Note that, in the preliminary conference version~\cite{KMC:22}, our interest is specified into output consensus control for sampled linear port-Hamiltonian systems (PHSs).

Based on the proposed discrete-time Krasovskii passivity concept, we study two design problems: stabilization and output consensus. Understanding the continuous-time stabilizing controllers in terms of passivity, we first characterize structures of sampled stabilizing controllers and derive conditions for the closed-loop stability under detectability assumptions. Next, moving to systems under unknown constant disturbances, we present sampled-data control methodologies for output consensus by interpreting continuous-time output consensus controllers from the passivity perspective. 
%Again, our results are natural extensions of the continuous-time ones, since the implicit midpoint discretization of continuous-time stabilizing controllers and output consensus controllers satisfy the requirements for the sampled discrete-time controllers, respectively. Namely, their discretized approximations achieve stabilization of sampled boost converters and output consensus of sampled islanded DC microgrids.

%A preliminary notion of Krasovskii passivity for sampled-data systems is in~\cite{KMC:22}, where the definition has been restricted to comply with output consensus control for sampled linear port-Hamiltonian systems (PHSs). 
The contribution of this paper is summarized as follows:
\begin{itemize}
\item We provide a novel notion of Krasovskii passivity for sampled discrete-time systems;

\item We establish the following implications for the relevant passivity concepts: incremental passivity $\implies$ Krasovskii passivity $\implies$ shifted passivity with respect to a suitable output function;

\item We develop sampled-data stabilizing control schemes for Krasovskii passive sampled discrete-time systems;

\item We develop sampled-data control methodologies for output consensus of Krasovskii passive sampled discrete-time systems under unknown constant disturbances.

%\item We show that implicit midpoint discretizations of Krasovskii passive (continuous-time) boost converters and islanded DC microgrids are again Krasovskii passive in the sampled discrete-time sense.
\end{itemize}

%%%%%%%%%%%%%%%%%%%%%%%%%%%%%%%%%%%%%%%%%%%%%%%%%%%%%%%%%%%%%%%%%%%%%
%%%%%%%%%%%%%%%%%%%%%%%%%%%%%%%%%%%%%%%%%%%%%%%%%%%%%%%%%%%%%%%%%%%%%

\smallskip

\subsubsection*{Organization}
The remainder of this paper is organized as follows. In Section~\ref{sec:pre}, we provide a motivating example. The goal is to illustrate the usefulness of Krasovskii passivity and its potential for sampled-data control. Then, we summarize Krasovskii-passivity-based control techniques for continuous-time systems to make the paper self-contained. In Section~\ref{sec:KP}, we introduce the concept of Krasovskii passivity for sampled discrete-time systems. Based on this Krasovskii passivity, we present sampled-data control techniques for stabilization and output consensus in Sections~\ref{sec:stab} and~\ref{sec:oc}, respectively. In Section~\ref{sec:ex}, the proposed sampled-data control techniques are illustrated by achieving voltage regulation (\emph{i.e.}, stabilization) for a DC network of boost converters and current sharing (\emph{i.e.}, output consensus control) for a DC network of buck converters with nonlinear loads. Section~\ref{sec:con} provides concluding remarks. Linear PHSs are used to aid in understanding each main result. Proofs are shown in the Appendices.

%%%%%%%%%%%%%%%%%%%%%%%%%%%%%%%%%%%%%%%%%%%%%%%%%%%%%%%%%%%%%%%%%%%%%
%%%%%%%%%%%%%%%%%%%%%%%%%%%%%%%%%%%%%%%%%%%%%%%%%%%%%%%%%%%%%%%%%%%%%

\smallskip

\subsubsection*{Notation}
The sets of real numbers, non-negative real numbers, integers, and non-negative integers are denoted by $\bR$, $\bR_+$, $\bZ$, and $\bZ_+$, respectively. 
The $n$-dimensional vector whose all components are $1$ is denoted by $\1_n$. 
Let $\circ$ denote the Hadamard product, \emph{i.e.}, given $x, y \in \bR^n$, $( x \circ  y) \in \bR^n$ is a vector with elements $( x \circ  y)_i := x_i y_i$ for all $i=1,\dots,n$. 
For $x \in \bR^n$, the diagonal matrix with $i$th diagonal element $x_i$, $i=1,\dots,n$ is denoted by ${\rm diag}\{x\}$.
The $n \times n$ identity matrix is denoted by $I_n$.
For a full column rank real matrix $A$, its Moore-Penrose inverse is denoted by $A^+ := (A^\top A)^{-1} A^\top$.
For $P \in \bR^{n \times n}$, $P \succ 0$ (resp. $P \succeq 0$) means that $P$ is symmetric and positive (resp. semi) definite.
For $x \in \bR^n$, its  Euclidean norm weighted by $P  \succ 0$ is denoted by $| x |_P := \sqrt{x^\top P x }$.
If $P = I_n$, this is simply described by $| x |$.

%%%%%%%%%%%%%%%%%%%%%%%%%%%%%%%%%%%%%%%%%%%%%%%%%%%%%%%%%%%%%%%%%%%%%
%%%%%%%%%%%%%%%%%%%%%%%%%%%%%%%%%%%%%%%%%%%%%%%%%%%%%%%%%%%%%%%%%%%%%

%%%%%%%%%%%%%%%%%%%%%%%%%%%%%%%%%%%%%%%%%%%%%%%%%%%%%%%%%%%%%%%%%%%%%
%%%%%%%%%%%%%%%%%%%%%%%%%%%%%%%%%%%%%%%%%%%%%%%%%%%%%%%%%%%%%%%%%%%%%

\section{Motivation and Review on Krasovskii passivity for continuous-time systems}\label{sec:pre}
In this section we first provide a motivating example to show the usefulness of Krasovskii passivity for control design and its potential for sampled-data control. Then, we recall the definition of Krasovskii passivity for continuous-time systems and the main control techniques based on it.

\subsection{Motivation}
Consider the average governing dynamic equations of a boost converter~ \cite{SPO:97}, \emph{i.e.},
\begin{align}\label{mex:csys}
\left\{\begin{array}{r@{}l}
L_s \dot I_s &{}= - (1 - u) V + V^*_s\\[0.5ex]
C \dot V &{}= (1 - u )I_s - G^*_l V
\end{array}\right.
\end{align}
with the states $I_s(t)$, $V(t) \in \bR$ denoting the average current and voltage, respectively, and the input $u(t) \in [0,1]$ denoting the duty ratio, where $L_s$, $C$, $V^*_s$, and $G^*_l$ are positive constants. Note that this is a bilinear system.

First, we verify its standard passivity by selecting the total energy as a storage function, \emph{i.e.},
\begin{align}\label{mex:storage}
S(I_s, V) := \frac{1}{2} (L_s I_s^2 + C V^2).
\end{align}
Its time derivative along the system trajectory satisfies
\begin{align*}
\dot S(I_s, V) =  - G^*_l V^2 + V^*_s I_s \le V^*_s I_s,
\end{align*}
which implies passivity with respect to the current $I_s$ and the constant voltage source $V^*_s$. However, $V^*_s$ is not a control input. Thus, this passivity property is not directly helpful for control design. 

Now, we consider replacing $(I_s, V)$ in~\eqref{mex:storage} with its time derivative $(\dot I_s, \dot V)$. 
Then, we employ the following storage function:
\begin{align}\label{mex:KPstorage}
S_K (\dot I_s, \dot V) := \frac{1}{2} (L_s \dot I_s^2 + C \dot V^2),
\end{align}
which satisfies 
\begin{align*}
\dot S_K (\dot I_s, \dot V) \le  \dot u (\dot I_s V - I_s \dot V).
\end{align*}
Based on this property, the  integral-type controller 
\begin{align}\label{mex:stab_con}
\dot u = K (u - u^*) - (\dot I_s V - I_s \dot V)
\end{align}
with $K > 0$ stabilizes the unique equilibrium $((I_s^*, V^*), u^*)$ (see~ \cite{CLK:19} and~ \cite{KCS:21} for further details). This approach is later formalized by introducing the notion of Krasovskii passivity as a tool for passivity-based integral-type control design~ \cite{KKS:21}.

In practice, continuous-time controllers are implemented after discretized approximation. However, their performance analyses are often neglected. In this paper, we are interested in studying this problem from the Krasovskii passivity perspective. In fact, it is reasonable to expect that a suitable discretization method will preserve the Krasovskii passivity of the boost converter~\eqref{mex:csys}. To confirm this, we calculate the implicit midpoint discretization of~\eqref{mex:csys}, \emph{i.e.},
\begin{align}\label{mex:sys}
&\left\{\begin{array}{r@{}l}
L_s \Delta_\delta I_{s,k} &{}= - (1 - u_k) \sigma_\delta V_k + V^*_s\\[0.5ex]
C \Delta_\delta V_k &{}= (1 - u_k ) \sigma_\delta I_{s,k} - G^*_l \sigma_\delta V_k,
\end{array}\right.
\end{align}
with
\begin{equation*}
\Delta_\delta I_{s,k} := \frac{I_{s, k+1} - I_{s,k}}{\delta}, 
\quad
\sigma_\delta I_{s,k} := \frac{I_{s, k+1} + I_{s,k}}{2}
\end{equation*}
and $\Delta_\delta V_k$ and $\sigma_\delta V_k$ similarly defined, while $\delta > 0$ and the subscript $k \in \bZ_+$ denote the sampling time and the time instant, respectively.
Based on the storage function~\eqref{mex:KPstorage} for Krasovskii passivity, we consider the following storage function:
\begin{align*}
S_K ( \Delta_\delta I_{s,k}, \Delta_\delta V_k )
:= \frac{1}{2} ( L_s (\Delta_\delta I_{s,k})^2 + C (\Delta_\delta V_k)^2 ),
\end{align*}
which satisfies
\begin{align*}
\Delta_\delta S_K ( \Delta_\delta I_{s,k}, \Delta_\delta V_k ) 
\le \Delta_\delta u_k (\Delta_\delta I_{s,k} \sigma_\delta V_k - \sigma_\delta I_k \Delta_\delta V_{s,k} ).
\end{align*}
%Note that a similar analysis is performed for a network of boost converters in Section~\ref{sec:ex_bst}, where numerical simulations are also provided.
The inequality above can be understood as Krasovskii passivity for the sampled boost converter~\eqref{mex:sys}.
Therefore, one can expect that it can be stabilized by the implicit midpoint discretization of the continuous-time stabilizing controller~\eqref{mex:stab_con}, which is further investigated in Section~\ref{sec:ex_bst} for a network of boost converters. 

Motivated by this example, the objective of this paper is to develop Krasovskii passivity theory for sampled discrete-time systems. To make the paper self-contained, the rest of this section is dedicated to reviewing results on Krasovskii passivity in the continuous-time case~\cite{KKS:21,KCF:22}. We first show its definition and applications to integral-type control design for stabilization and output consensus.

%%%%%%%%%%%%%%%%%%%%%%%%%%%%%%%%%%%%%%%%%%%%%%%%%%%%%%%%%%%%%%%%%%%%%
%%%%%%%%%%%%%%%%%%%%%%%%%%%%%%%%%%%%%%%%%%%%%%%%%%%%%%%%%%%%%%%%%%%%%

\subsection{Definition}
Consider a continuous-time nonlinear system:
\begin{align}\label{ceq:sys}
\dot x = f (x, u),
\end{align}
with state $x(t) \in \bR^n$, input $u(t) \in \bR^m$, and class $C^1$ mapping $f :\bR^n \times \bR^m \to \bR^n$.
Krasovskii passivity is introduced under the existence of an equilibrium point.
\smallskip
\begin{secasm}\label{casm:eq}
For the system~\eqref{ceq:sys}, the set 
\begin{align*}
\cE :=\{(x^*, u^*) \in \bR^n \times \bR^m: f (x^*, u^*)=0\}
\end{align*}
is not empty. \red
\end{secasm}
\smallskip
Krasovskii passivity is defined as passivity of the so-called extended system~\cite{Schaft:82}:
\begin{align}\label{ceq:sys_ex}
\left\{\begin{array}{r@{}l}
\dot x &{}= f (x, u) \\[0.5ex]
\dot u &{}= v \\[0.5ex]
z &{}= \hat q (x, u) :=  q (x, u, \dot x),
\end{array}\right.
\end{align}
where $q :\bR^n \times \bR^m \times \bR^n \to \bR^m$ is continuous, and $\hat q(x,u)$ is defined by substituting $\dot x = f(x,u)$ into $q(x, u, \dot x)$.
This yields a system with the extended state $(x, u) \in \bR^n \times \bR^m$, input $v \in \bR^m$, and output $z \in \bR^m$.
Now, we are ready to state the definition of Krasovskii passivity.
\smallskip
\begin{secdefn}
Under Assumption~\ref{casm:eq}, the system~\eqref{ceq:sys} is said to be \emph{Krasovskii passive} at~$(x^*, u^*) \in \cE$ if for its extended system~\eqref{ceq:sys_ex}, there exist a class~$C^1$ function~$S_K:\bR^n  \times \bR^m \times \bR^n \to \bR_+$ and a continuous function~$W_K:\bR^n \times \bR^m \times \bR^n \to \bR_+$ such that $S_K(x^*, u^*, 0 )=0$, and $\hat S_K(x, u) := S_K(x, u, \dot x )$ satisfies
\begin{align*}
\frac{\partial \hat S_K(x, u)}{\partial x} f(x, u) + \frac{\partial \hat S_K(x, u)}{\partial u} v 
\le - W_K(x, u, \dot x ) + v^\top z
\end{align*}
for all~$(x, u) \in \bR^n \times \bR^m$ and~$v \in \bR^m$. Moreover, the system is said to be \emph{strictly Krasovskii passive} if there exists a class~$C^1$ function~$h:\bR^n \to \bR^m$ such that $z = \dot y$ for $y := h(x)$, and $W_K(x, u, \dot x) = 0$ implies $\dot x = 0$ for all $(x,u) \in \bR^n \times \bR^m$.
\red
\end{secdefn}\smallskip

Strict Krasovskii passivity is stronger than Krasovskii passivity in two aspects. First, the Krasovskii passive output is the time derivative of $y$. Second, $W_K$ is positive definite with respect to $\dot x$. In contrast, an equilibrium point is not necessarily to exist. Below, we respectively employ Krasovskii passivity and strict Krasovskii passivity for stabilizing control and output consensus control, where the existence of an equilibrium point is not required for output consensus.

%%%%%%%%%%%%%%%%%%%%%%%%%%%%%%%%%%%%%%%%%%%%%%%%%%%%%%%%%%%%%%%%%%%%%
%%%%%%%%%%%%%%%%%%%%%%%%%%%%%%%%%%%%%%%%%%%%%%%%%%%%%%%%%%%%%%%%%%%%%

\subsection{Stabilizing Control}
In this subsection, we recall results in~\cite{KKS:21} for Krasovskii-passivity-based stabilizing control, studied under the following detectability assumption for the extended system~\eqref{ceq:sys_ex}.
\smallskip
\begin{secdefn}
Under Assumption~\ref{casm:eq}, the extended system~\eqref{ceq:sys_ex} is said to be \emph{detectable} at~$(x^*, u^*) \in \cE$ if 
\begin{align*}
&v(\cdot) = 0 \mbox{\quad and\quad} z(\cdot )=0\\
&\quad \implies \quad \lim_{t \to \infty} (x(t), u(t)) = (x^*, u^*)
\end{align*}
holds.
\red
\end{secdefn}\smallskip

Krasovskii passivity can be understood as passivity of the extended system~\eqref{ceq:sys_ex}, and thus a stabilizing controller is designed via it. Namely, we consider the following integral-type controller:
\begin{align}\label{ceq:con_stab}
K_1 \dot u =  - K_2 (u - u^*) - z,
\end{align}
where~$K_1, K_2 \succ 0$ ($K_1, K_2 \in \bR^{m \times m}$) are free tuning parameters. 
Note that the controller has the following Krasovskii passivity property while this is not mentioned in~\cite{KKS:21}:
\begin{align*}
\frac{1}{2} \frac{d}{dt}  |u - u^*|_{K_2}^2 
=  - |\dot u |_{K_1}^2 - \dot u^\top z.
\end{align*}
Based on this property, we have the following proposition for closed-loop stability, which is a slight modification of the original result~\cite[Corollary 3.2]{KKS:21}.
\smallskip
\begin{secprop}\label{cprop:stab}
Consider the closed-loop system consisting of a Krasovskii passive system~\eqref{ceq:sys} at $(x^*, u^*) \in \cE$ and a controller~\eqref{ceq:con_stab}.
Then, the following two statements hold:
\begin{itemize}
\item[(a)] If the closed-loop system is positively invariant on a compact set $\Omega \subset \bR^n \times \bR^m$ containing $(x^*, u^*)$, then  any trajectory starting from~$\Omega$ converges to the largest invariant set contained in
\begin{align*}
&\{ (x, u) \in \Omega  : \hat z = 0, W_K = 0 \}\\
&\ \ \, \hat z := K_2 (u- u^*) + \hat q (x, u)
\end{align*}

\item[(b)] If the extended system~\eqref{ceq:sys_ex} with the output~$(\hat z, W_K)$ is detectable at~$(x^*,u^*)$, then~$(x^*, u^*)$ is asymptotically stable.
\red
\end{itemize}
\end{secprop}\smallskip

For the controller~\eqref{ceq:con_stab}, we only specify $u^*$, and information of $x^*$ is not required. In general, $x^*$ depends on system parameters. Thus, this structure is helpful for robust control design as shown in~\cite{KKS:21} for the DC-Zeta converter, in~\cite{CLK:19} for the boost converter, and in~\cite{KCS:21} for general Brayton-Moser systems.

%%%%%%%%%%%%%%%%%%%%%%%%%%%%%%%%%%%%%%%%%%%%%%%%%%%%%%%%%%%%%%%%%%%%%
%%%%%%%%%%%%%%%%%%%%%%%%%%%%%%%%%%%%%%%%%%%%%%%%%%%%%%%%%%%%%%%%%%%%%

\subsection{Output Consensus Control}
In the context of network control, strict Krasovskii passivity has been applied to output consensus control under external unknown constant disturbance~\cite{KCF:22}, which is summarized in this subsection.

Consider the following system under unknown constant disturbance:
\begin{align}\label{ceq:sys_d}
\left\{\begin{array}{r@{}l}
\dot x &{}= g (x, u, d)\\[0.5ex]
y &{}=  h (x, d)
\end{array}\right.
\end{align}
with the output $y(t) \in \bR^m$ and unknown constant disturbance $d \in \bR^r$, where $g  :\bR^n \times \bR^m \times \bR^r \to \bR^n$ and $h  :\bR^n \times \bR^r \to \bR^m$ are of class $C^1$ such that $\partial g/\partial u$ is of full column rank at each $(x, u, d) \in \bR^n \times \bR^m \times \bR^r$. The rank property is imposed to guarantee that $\dot x=0$ (and $\ddot x = 0$) implies $\dot u = 0$.

For output consensus, we consider the following integral-type controller:
\begin{align}\label{ceq:con_oc}
\left\{\begin{array}{r@{}l}
\dot \xi &{}= E^\top y\\[0.5ex]
u &{}= -E \xi
\end{array}\right.
\end{align}
with the controller state $\xi(t) \in \bR^p$, where $E \in \bR^{m \times p}$ is such that ${\rm rank} \; E^\top = m-1$ and $E^\top \1_m = 0$. For instance, $E$ is the incidence matrix of an undirected connected communication graph. Note that the controller has the following Krasovskii passivity property while this is not mentioned in~\cite{KCF:22}:
\begin{align*}
\frac{1}{2} \frac{d}{dt} |E^\top y|^2 = y^\top E E^\top \dot y = \dot \xi^\top E^\top \dot y = -\dot u^\top \dot y.
\end{align*}

This controller achieves output consensus for strictly Krasovskii passive systems \cite[Theorem 3.2]{KCF:22}.
\smallskip
\begin{secprop}
Given $d \in \bR^r$, suppose that the closed-loop system consisting of a strictly Krasovskii passive system~\eqref{ceq:sys_d} and a controller~\eqref{ceq:con_oc} is positively invariant on a compact set $\Omega \subset \bR^n \times \bR^p$. Then, for each $(x(0), \xi(0)) \in \Omega$, there exists $\alpha: \bR \to \bR$ such that $\lim_{t \to \infty} ( y (t) - \alpha (t) \1_m) = 0$.~\red
\end{secprop}

%%%%%%%%%%%%%%%%%%%%%%%%%%%%%%%%%%%%%%%%%%%%%%%%%%%%%%%%%%%%%%%%%%%%%
%%%%%%%%%%%%%%%%%%%%%%%%%%%%%%%%%%%%%%%%%%%%%%%%%%%%%%%%%%%%%%%%%%%%%

%%%%%%%%%%%%%%%%%%%%%%%%%%%%%%%%%%%%%%%%%%%%%%%%%%%%%%%%%%%%%%%%%%%%%
%%%%%%%%%%%%%%%%%%%%%%%%%%%%%%%%%%%%%%%%%%%%%%%%%%%%%%%%%%%%%%%%%%%%%

\section{Krasovskii Passivity for Sampled Discrete-time Systems}\label{sec:KP}
This paper explores Krasovskii-passivity-based control schemes for sampling data. We first investigate the notion of Krasovskii passivity for discrete-time systems arising from temporal discretization. To better understand this passivity notion, we next study its connection with other two passivity notions: incremental passivity and shifted passivity. In the following two sections, we study Krasovskii-passivity-based control design for stabilization and output consensus. In each section, we employ linear port-Hamiltonian systems (PHSs) to illustrate the proposed concept and control design.

%%%%%%%%%%%%%%%%%%%%%%%%%%%%%%%%%%%%%%%%%%%%%%%%%%%%%%%%%%%%%%%%%%%%%
%%%%%%%%%%%%%%%%%%%%%%%%%%%%%%%%%%%%%%%%%%%%%%%%%%%%%%%%%%%%%%%%%%%%%
% \begin{figure}[tb]
%     \centering
%     \includegraphics[width=\linewidth]{Block_scheme.pdf}
%     \caption{Sampled-data scheme}
%     \label{fig:scheme}
% \end{figure}

\begin{figure}[t]
\begin{center}
\begin{tikzpicture}[scale=1, transform shape]

\node [text width=7cm, text height=1.5cm, fill=white, dashed, draw=gray, rectangle, rounded corners, text centered] at (3.15,0.2){ };
\draw (0.75,1.3) node
{Sampled model};

\bXInput{A}

\bXStyleBloc{rounded corners,fill=blue!20,text=blue}
\bXBloc[0]{zoh}{Holder}{A}

\bXStyleBloc{rounded corners,fill=black!10,text=blue}
\bXBloc[4]{plant}{Plant}{zoh}
\bXLink[$u(t)$]{zoh}{plant}

\bXStyleBloc{rounded corners,fill=blue!20,text=blue}
\bXBloc[4]{sampler}{Sampler}{plant}
\bXLink[$x(t)$]{plant}{sampler}

\bXOutput[2]{Z}{sampler}
\bXBranchy[5]{Z}{shift}
\bXLineStyle{-}
\bXLinkxy{sampler}{shift}

\bXStyleBloc{rounded corners,fill=blue!20,text=blue}
\bXBloc[-13.45]{controller}{Controller}{shift}
\bXOutput[-13.2]{shiftC}{controller}
\bXDefaultLineStyle
\bXLinkyx{Z}{controller}
\bXLink[$x(k\delta)$]{shift}{controller}
\bXLineStyle{-}
\bXLink[$u(k\delta)$]{controller}{shiftC}
\bXDefaultLineStyle
\bXLinkyx{shiftC}{zoh}

%\bXLineStyle{blue, dotted}
%\bXBranchy[-4]{sensor}{parameters}
%\bXLink{parameters}{sensor}
%{\large{
%\bXLabelStyle{text=blue}
%\bXLinkName[0.5]{parameters}{$\delta$}}}
%\bXDefaultLineStyle
%
%{\bXLineStyle{black, dashed}
%\bXOutput[5]{Z}{sensor}
%{\large{
%\bXLink[$\sigma(x(t_k))$]{sensor}{Z}}}
%
%\bXBranchy[9]{sensor-Z}{shift}
%\begin{small}
%\tikzstyle{bXStyleBlocPotato}=[dashed, cloud puffs=7, fill = blue!60!gray!2]
%\bXBlocPotato[-25]{n2}{Network}{shift}
%\end{small}
%\bXLinkyx{sensor-Z}{n2}
%\bXLinkxy{n2}{controller}}
\end{tikzpicture}
\caption{Sampled-data scheme.}
\label{fig:scheme}
\end{center}
\end{figure}
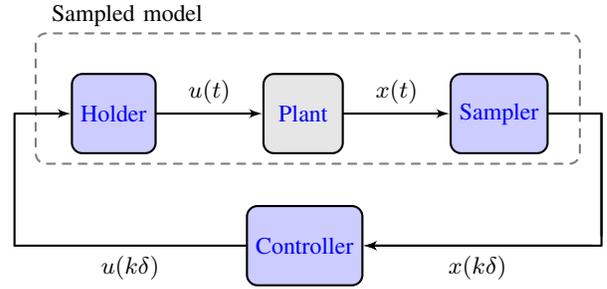
%%%%%%%%%%%%%%%%%%%%%%%%%%%%%%%%%%%%%%%%%%%%%%%%%%%%%%%%%%%%%%%%%%%%%
%%%%%%%%%%%%%%%%%%%%%%%%%%%%%%%%%%%%%%%%%%%%%%%%%%%%%%%%%%%%%%%%%%%%%

\subsection{Sampling and numerical integration}
A typical approach to computer control is based on the sample-and-hold technique~ \cite{CF:12,Sontag:13}, as in Fig.~\ref{fig:scheme}. The value of the state $x(t)$ is measured only at sampling times (discrete instants) $k\delta$ for each $\delta > 0$ and every $k\in \bZ$. The input $u(t)$ is generated by the zero-order hold that is used to take discrete signals at time $k\delta$ and to hold the value constant in the time interval $[k\delta, k\delta+ \delta )$, which is fed into the plant as a control. Formally, we consider locally integrable $u:[0,\infty)\to \mathbb{R}^m$. Then, for a sampling period $\delta >0$ and any $k\in \bZ$, we have $u(t)=u(k\delta)=:u_k$ for all $t\in [k\delta, k\delta+\delta )$. Then, applying temporal discretization, \emph{e.g.}~\cite{HLW:06}, to the continuous-time system~\eqref{ceq:sys} yields the following sampled discrete-time system:
\begin{align}\label{eq:sys}
\Delta_\delta x_k := \frac{x_{k+1} -x_k}{\delta} = f_\delta (\sigma_\delta x_k, u_k),
\end{align}
where $f_\delta : \bR^n \times \bR^m \to \bR^n$ is continuous at each $\delta$, and $\sigma_\delta$ denotes a time-shift operator mapping $x_k$ to some continuous function of $(x_k, x_{k+1})$ at each $\delta$; also see Example~\ref{ex:discretization} below.
\smallskip
\begin{secrem}
In general, temporal discretization of the continuous-time system~\eqref{ceq:sys} may depend also on $x_k$. Moreover, $\Delta_\delta x_k$ and $\sigma_\delta x_k$ may depend on further future states $x_{k+2}$, $x_{k+3}$, and so forth. Our results can be extended to such cases. However, to deliver the essence of our approach while keeping notation concise, we focus on a sampled discrete-time system in the form of~\eqref{eq:sys}.
\red
\end{secrem}

\smallskip

Throughout this paper, we suppose that~\eqref{eq:sys} is well defined.

\begin{secsasm}\label{sasm}
For each sequence $\{u_k\}_{k \in \bZ} \in (\bR^m)^{\bZ}$, a sequence $\{x_k\}_{k \in \bZ} \in (\bR^n)^{\bZ}$ satisfying~\eqref{eq:sys} exists and is unique.
\red
\end{secsasm}\smallskip

For instance, Standing Assumption~\ref{sasm} holds if a system~\eqref{eq:sys} admits an explicit form:
\begin{align}\label{eq:sys_exp}
x_{k+1} = F_\delta (x_k, u_k),
\end{align}
where $F_\delta:\bR^n \times \bR^m \to \bR^n$ is continuous at each $\delta$.
This is a system such that for each sequence $\{u_k\}_{k \in \bZ} \in (\bR^m)^{\bZ}$, a sequence $\{x_k\}_{k \in \bZ} \in (\bR^n)^{\bZ}$ generated by~\eqref{eq:sys_exp} is a unique solution to~\eqref{eq:sys}.

In this paper, we apply time-shift operators to functions also. For a function $a(\{x_k, u_k\}_{k \in \bZ})$, they are respectively defined as
\begin{align*}
\Delta_\delta a(\{x_k, u_k\}_{k \in \bZ}) &: = (\Delta_\delta z_k)|_{z_k =a(\{x_k, u_k\}_{k \in \bZ})}, \\[1ex]
\sigma_\delta a(\{x_k, u_k\}_{k \in \bZ}) &: =  (\sigma_\delta z_k)|_{z_k =a(\{x_k, u_k\}_{k \in \bZ})}.
\end{align*}
For the sake of notional simplicity, the self-composition of $\Delta_\delta$ is denoted by $\Delta_\delta^1:= \Delta_\delta$ and {$\Delta_\delta^{i+1} := \Delta_\delta \bar \circ \Delta_\delta^i$}, $i=1,2,\dots$, {where $\bar \circ$ denotes the composition operator.} The operator  $\sigma_\delta^i$, $i=1,2,\dots$ is defined similarly.
Next, to be consistent with continuous-time dynamics, we assume for any constant $c$ that $\sigma_\delta$ and $c_k := c$, $k \in \bZ$ satisfy
\begin{align}\label{eq:const}
\sigma_\delta  c_k  = c, \quad \forall k \in \bZ.
\end{align}

The class of sampled discrete-time systems~\eqref{eq:sys} contains one obtained by the forward Euler and implicit midpoint methods, which is explained in the following example.
\smallskip
\begin{secex}\label{ex:discretization}
For the forward Euler method, we have
\begin{align*}
f_\delta (\sigma_\delta x_k, u_k) &= f (\sigma_\delta x_k, u_k),\\
\sigma_\delta x_k & = x_k.
\end{align*}
For the implicit midpoint method, we have
\begin{subequations}\label{eq:IMM}
\begin{align}
f_\delta (\sigma_\delta x_k, u_k) &= f (\sigma_\delta x_k, u_k),
\label{eq1:IMM}\\
%%%%%%%%
%%%%%%%%
\sigma_\delta x_k & = \frac{x_k + x_{k+1}}{2}.
\label{eq2:IMM}
\end{align}
\end{subequations}
Both methods satisfy~\eqref{eq:const}.

Next, we apply $\Delta_\delta$ to the quadratic function $|x_k|^2/2$. It follows that
\begin{align*}
 {\Delta_\delta \left( \frac{|x_k|^2}{2} \right)}
&= \frac{|x_{k+1}|^2 - |x_k|^2}{2\delta}\\
&= \frac{(x_{k+1} -x_k)^\top}{\delta} \frac{(x_k + x_{k+1})}{2}\\
&= (\Delta_\delta x_k)^\top  \sigma_\delta x_k.
\end{align*}
Similar calculations are later used for the analysis of linear PHSs.~\red
\end{secex}

%%%%%%%%%%%%%%%%%%%%%%%%%%%%%%%%%%%%%%%%%%%%%%%%%%%%%%%%%%%%%%%%%%%%%
%%%%%%%%%%%%%%%%%%%%%%%%%%%%%%%%%%%%%%%%%%%%%%%%%%%%%%%%%%%%%%%%%%%%%

\subsection{Definition}

In this paper, we investigate Krasovskii passivity for a sampled discrete-time system~\eqref{eq:sys} by assuming that it has an equilibrium point in the discrete-time sense.
\smallskip
\begin{secasm}\label{asm:KP}
For the system~\eqref{eq:sys}, the set
\begin{align*}
\cE :=\{(x^*, u^*) \in \bR^n \times \bR^m:  x^* = f_\delta (\sigma_\delta x^*, u^*)\}
\end{align*}
is not empty. 
\red
\end{secasm}\smallskip

In the continuous-time case, Krasovskii passivity is defined as passivity of the extended system~\eqref{ceq:sys_ex} by viewing $\dot u$ as the new input. Mimicking this procedure, we introduce an extended system of the discrete-time system~\eqref{eq:sys} as follows:
\begin{align}\label{eq:sys_ex}
\left\{\begin{array}{r@{}l}
\Delta_\delta x_k &{}= f_\delta (\sigma_\delta x_k, u_k)\\[0.5ex]
\Delta_\delta u_k &{}= v_k\\[0.5ex]
z_k &{}=  q_\delta ( \sigma_\delta x_k, u_k, \Delta_\delta \sigma_\delta x_k),
\end{array}\right.
\end{align}
where $q_\delta :\bR^n \times \bR^m \times \bR^n \to \bR^m$ is a continuous function at each $\delta$.
This is a system with the extended state $(x_k,u_k)$, input $v_k \in \bR^m$, and output $z_k \in \bR^m$.
Under Standing Assumption~\ref{sasm}, this extended system is well defined also.

Now, we are ready to introduce the concept of Krasovskii passivity for sampled discrete-time systems. 
\smallskip
\begin{secdefn}\label{def:KP}
Under Assumption~\ref{asm:KP}, the discrete-time system~\eqref{eq:sys} is said to be \emph{Krasovskii passive} at~$(x^*, u^*) \in \cE$ if for its extended system~\eqref{eq:sys_ex}, there exist continuous functions~$S_K, W_k: \bR^n \times \bR^m \times \bR^n \to \bR_+$ such that~$S_K( x^*, u^*, 0)=0$ and
\begin{subequations}\label{eq:KP}
\begin{align}\label{eq1:KP}
&\Delta_\delta S_K( x_k, u_k, \Delta_\delta x_k) \nonumber\\
&\le - W_K (\sigma_\delta x_k, u_k, \Delta_\delta \sigma_\delta x_k ) + v_k^\top z_k
\end{align}
for all~$(x_k, u_k) \in \bR^n \times \bR^m$ and~$v_k \in \bR^m$.
Moreover, the system is said to be \emph{strictly Krasovskii passive} if there exists a continuous function $h_\delta : \bR^n \to \bR^m$ such that $z_k = \Delta_\delta y_k$ for $y_k := h_\delta (\sigma_\delta x_k)$, and
\begin{align}\label{eq2:KP}
W_K( \sigma_\delta x_k, u_k, \Delta_\delta \sigma_\delta x_k )=0 \quad \implies \quad \Delta_\delta \sigma_\delta x_k =0
\end{align}
\end{subequations}
for all~$(x_k, u_k) \in \bR^n \times \bR^m$. 
The function $S_K$ is referred to as a \emph{storage function}.
\red
\end{secdefn}\smallskip

Similarly to the continuous-time case, strict Krasovskii passivity is stronger than Krasovskii passivity in two aspects. First, the Krasovskii passive output is the discretization of the time derivative of $y_k$. Second, $W_K$ is positive definite with respect to $\Delta_\delta \sigma_\delta x_k$. In contrast, an equilibrium point is not necessarily to exist. As for the continuous-time case, we develop sampled stabilizing controllers based on Krasovskii passivity and sampled output consensus controllers based on strict Krasovskii passivity, where the latter control design does not require the existence of an equilibrium point.

\smallskip
\begin{secrem}\label{rem2:KP}
In sampled discretization considered in this paper, $u (\delta k)$ is replaced with $u_k$ by viewing $u(t)$ as the control input. Accordingly, $\lim_{\tau \to 0^+} (u (\delta k + \tau) -u (\delta k))/\tau$ is replaced with $\Delta_\delta u_k$. In contrast, a function of the state $x(\delta k)$ (resp. $\lim_{\tau \to 0^+} (x (\delta k + \tau) - x (\delta k))/\tau$) is replaced with some function of $\sigma_\delta x_k$ (resp. $\Delta_\delta \sigma_\delta x_k$). In the continuous-time case, the extended system~\eqref{ceq:sys_ex} can be viewed as a system with the states $x$ and $u$. Thus, one may consider another sampled discretization by replacing a function of $(x, u)$ with some function of $(\sigma_\delta x_k, \sigma_\delta u_k)$. Our results can be extended to such a case; see the preliminary version for linear PHSs~\cite[Section 3.3]{KMC:22}.~\red
\end{secrem}

%%%%%%%%%%%%%%%%%%%%%%%%%%%%%%%%%%%%%%%%%%%%%%%%%%%%%%%%%%%%%%%%%%%%%
%%%%%%%%%%%%%%%%%%%%%%%%%%%%%%%%%%%%%%%%%%%%%%%%%%%%%%%%%%%%%%%%%%%%%

\subsection{Illustration by Linear Port-Hamiltonian Systems}\label{ex1:LPH}
We illustrate Definition~\ref{def:KP} for Krasovskii passivity of a sampled discrete-time system through analysis of a linear PHS:
\begin{align}\label{clPHeq:sys}
\left\{\begin{array}{r@{}l}
\dot x &{}=   (J - R) H x  + B u \\[0.5ex]
y &{}=  B^\top H x,
\end{array}\right.
\end{align}
where $H , R \succeq 0$ ($H , R\in \bR^{n \times n}$) and $J = -J^\top \in \bR^{n \times n}$.
The implicit midpoint method~\eqref{eq:IMM} provides the following discrete-time system:
\begin{align}\label{lPHeq:sys}
\left\{\begin{array}{r@{}l}
\Delta_\delta x_k &{}= (J - R) H \sigma_\delta x_k + B u_k \\[0.5ex]
y_k &{}=  B^\top H \sigma_\delta x_k.
\end{array}\right.
\end{align}
If the sampling period $\delta > 0$ is selected such that $I_n/\delta - (J - R) H/2$ is invertible, Standing Assumption~\ref{sasm} holds.

The obtained discrete-time PHS is Krasovskii passive.
\smallskip
\begin{secthm}\label{lPHthm:KP}
A discrete-time PHS~\eqref{lPHeq:sys} is Krasovskii passive with respect to the following storage function:
\begin{align}\label{IPHeq:storage}
S_K (  \Delta_\delta x_k) : =  \frac{1}{2} |\Delta_\delta x_k|_H^2.
\end{align}
Moreover, this is strictly Krasovskii passive if $H, R \succ 0$.
\end{secthm}\smallskip

\begin{proof}
As a preliminary step, we apply $\Delta_\delta$ to the system~\eqref{lPHeq:sys}, yielding
\begin{align}\label{lPHeq:sys_ex}
\left\{\begin{array}{r@{}l}
\Delta_\delta^2 x_k = {}& (J - R) H \Delta_\delta \sigma_\delta x_k + B \Delta_\delta u_k \\[0.5ex]
\Delta_\delta u_k = {}& v_k\\[0.5ex]
z_k := {}&\Delta_\delta y_k =  B^\top H \Delta_\delta \sigma_\delta x_k.
\end{array}\right.
\end{align}
We show Krasovskii passivity with respect to the following storage function:
\begin{align*}
S_K (  \Delta_\delta x_k) : =  \frac{1}{2} |\Delta_\delta x_k|_H^2,
\end{align*}
which is zero if $\Delta_\delta x_k = 0$. 
It follows from~\eqref{lPHeq:sys_ex} that
\begin{align}\label{lPHeq:KP}
\Delta_\delta S_K ( \Delta_\delta x_k) 
&= \frac{|\Delta_\delta x_{k+1}|_H^2 - |\Delta_\delta x_k|_H^2}{2 \delta}  \nonumber\\
& = \frac{(\Delta_\delta x_{k+1} + \Delta_\delta x_k)^\top }{2} H \frac{\Delta_\delta x_{k+1} - \Delta_\delta x_k}{\delta}  \nonumber\\
& = (\Delta_\delta \sigma_\delta x_k)^\top H \Delta_\delta^2 x_k  \nonumber\\
& = -|H \Delta_\delta \sigma_\delta x_k|_R^2 + (\Delta_\delta \sigma_\delta x_k)^\top  H B v_k \nonumber\\
& = -|H \Delta_\delta \sigma_\delta x_k|_R^2 +  v_k^\top z_k.
\end{align}
Since $R \succeq 0$, the system is Krasovskii passive.
Finally, we show strict Krasovskii passivity. From~\eqref{lPHeq:sys_ex}, $z_k = \Delta_\delta y_k$. Next,~\eqref{eq2:KP} holds for $W_K(\Delta_\delta \sigma_\delta x_k) := |H \Delta_\delta \sigma_\delta x_k|_R^2$ if $H, R \succ 0$.
\end{proof}

%%%%%%%%%%%%%%%%%%%%%%%%%%%%%%%%%%%%%%%%%%%%%%%%%%%%%%%%%%%%%%%%%%%%%
%%%%%%%%%%%%%%%%%%%%%%%%%%%%%%%%%%%%%%%%%%%%%%%%%%%%%%%%%%%%%%%%%%%%%

\subsection{Relations with Incremental Passivity and Shifted Passivity}
At the end of this section, we investigate the connections of Krasovskii passivity with incremental passivity and shifted passivity. The objective is to show the implications in Fig.~\ref{fig:KPrelation}.

%%%%%%%%%%%%%%%%%%%%%%%%%%%%%%%%%%%%%%%%%%%%%%%%%%%%%%%%%%%%%%%%%%%%%
%%%%%%%%%%%%%%%%%%%%%%%%%%%%%%%%%%%%%%%%%%%%%%%%%%%%%%%%%%%%%%%%%%%%%

\begin{figure}[t!]
    \centering
\includegraphics[width=\linewidth]{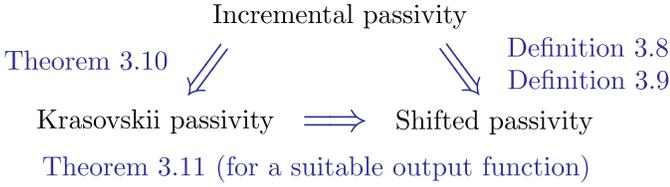}
    \caption{Relations among Krasovskii Passivity, Incremental Passivity, and Shifted Passivity.}
    \label{fig:KPrelation}
\end{figure}

%%%%%%%%%%%%%%%%%%%%%%%%%%%%%%%%%%%%%%%%%%%%%%%%%%%%%%%%%%%%%%%%%%%%%
%%%%%%%%%%%%%%%%%%%%%%%%%%%%%%%%%%%%%%%%%%%%%%%%%%%%%%%%%%%%%%%%%%%%%

Throughout this subsection, we assume that a system~\eqref{eq:sys} admits an explicit form~\eqref{eq:sys_exp} with output $y_k = H_\delta (x_k, u_k)$:
\begin{align}\label{eq:sys_exp_out}
\left\{\begin{array}{r@{}l}
x_{k+1} &{}= F_\delta (x_k, u_k)\\[0.5ex]
y_k &{}= H_\delta (x_k, u_k)
\end{array}\right.
\end{align}
where $H_\delta: \bR^n \times \bR^m \to \bR^m$ is continuous at each $\delta$. From~\eqref{eq:sys_ex}, its extended system is
\begin{align}\label{eq:sys_ex_exp}
\left\{\begin{array}{r@{}l}
x_{k+1} &{}= F_\delta (x_k, u_k)\\[0.5ex]
u_{k+1} &{}= u_k + \delta v_k\\[0.5ex]
z_k &{} :=  \Delta_\delta y_k  = \Delta_\delta H_\delta (x_k, u_k).
\end{array}\right.
\end{align}
For the sake of notional simplicity, we rewrite the storage function for Krasovskii passivity as
\begin{align}\label{eq:KP_storage_exp}
\hat S_K(x_k, u_k) := S_K( x_k, u_k, \Delta_\delta x_k).
\end{align}

Incremental passivity and shifted passivity are originally defined for continuous-time systems, \emph{e.g.},~\cite{PM:08,Schaft:13,Schaft:17}. These concepts can readily be transferred to sampled discrete-time systems. First, we  define incremental passivity by introducing an auxiliary system of~\eqref{eq:sys_exp_out} and its copy: 
\begin{align}\label{eq:sys_ax}
\left\{\begin{array}{r@{}l}
x_{k+1} &{}= F_\delta (x_k, u_k)\\[0.5ex]
x'_{k+1} &{}= F_\delta (x'_k, u'_k)\\[0.5ex]
y_k &{}= H_\delta (x_k, u_k)\\[0.5ex]
y'_k &{}= H_\delta (x'_k, u'_k),
\end{array}\right.
\end{align}

We define incremental passivity for discrete-time systems as follows.
\smallskip
\begin{secdefn}
The discrete-time system~\eqref{eq:sys_exp_out} is said to be \emph{incrementally passive} if for its auxiliary system~\eqref{eq:sys_ax}, there exists a continuous function~$S_I: \bR^n \times \bR^n \to \bR_+$ such that~$S_I(x, x)=0$ for all $x \in \bR^n$, and
\begin{align}\label{eq:IP}
\Delta_\delta S_I( x_k, x'_k) \le (u_k - u'_k)^\top (y_k - y'_k),
\end{align}
for all~$((x_k, u_k), (x'_k, u'_k)) \in (\bR^n \times \bR^m) \times (\bR^n \times \bR^m)$. 
\red
\end{secdefn}\smallskip

Next, shifted passivity is defined by fixing a trajectory $(x'_k, u'_k)$ of incremental passivity into an equilibrium point $(x^*, u^*) \in \cE$ as follows.
\smallskip
\begin{secdefn}
Under Assumption~\ref{asm:KP}, the discrete-time system~\eqref{eq:sys_exp_out} is said to be \emph{shifted passive} at $(x^*, u^*) \in \cE$ if there exists a continuous function~$S_S: \bR^n \to \bR_+$ such that~$S_S(x^*)=0$ and
\begin{align*}
\Delta_\delta S_S(x_k) \le (u_k - u^*)^\top (y_k - y^*),
\end{align*}
for all~$(x_k, u_k) \in \bR^n \times \bR^m$, where $y^* := H_\delta (x^*, u^*)$. 
\red
\end{secdefn}\smallskip

From their definitions, it is clear that incremental passivity implies shifted passivity. 
We further connect them with Krasovskii passivity as in Fig.~\ref{fig:KPrelation}.
First, we show that incremental passivity implies Krasovskii passivity. 
\smallskip
\begin{secthm}\label{thm:IP}
Under Assumption~\ref{asm:KP}, an incrementally passive discrete-time system~\eqref{eq:sys_exp_out} is Krasovskii passive at each $(x^*, u^*) \in \cE$.
\end{secthm}\smallskip
\begin{proof}
The proof is in Appendix~\ref{app:IP}.
\end{proof}\smallskip 

Next, we show that Krasovskii passivity implies shifted passivity with respect to a suitable output function.
\smallskip
\begin{secthm}\label{thm:SP}
Under Assumption~\ref{asm:KP}, a Krasovskii passive discrete-time system~\eqref{eq:sys_exp_out} at $(x^*, u^*) \in \cE$ is shifted passive at the same $(x^*, u^*)$ with respect to the following output:
\begin{align}\label{eq:SP_out}
y_k := \delta \int_0^1 \left. \frac{\partial^\top \hat S_K( F_\delta (x_k, u), u^*)}{\partial u} \right|_{u = s u_k + (1-s) u^*} ds.
\end{align}
\vspace{1ex}
\end{secthm}\smallskip
\begin{proof}
The proof is in Appendix~\ref{app:SP}.
\end{proof}\smallskip

Also in the continuous-time case, the implication: Krasovskii passivity $\implies$ shifted passivity has been shown \cite[Theorem 2.16]{KKS:21} for input-affine systems. While a discrete-time system~\eqref{eq:sys_exp_out} is not input-affine, the proof idea of Theorem~\ref{thm:SP} is partly borrowed from \cite[Theorem 2.16]{KKS:21}. In fact, Theorem~\ref{thm:SP} gives a new insight for the continuous-time case that the implication can be shown without assuming a system to be input-affine.

The implication of Theorem~\ref{thm:IP}: incremental passivity $\implies$ Krasovskii passivity has not been investigated in the continuous-time case. Its proof is based on discrete-time nature. From Theorem~\ref{thm:IP}, Krasovskii-passivity-based control techniques proposed below are applicable to incrementally passive systems also.

%%%%%%%%%%%%%%%%%%%%%%%%%%%%%%%%%%%%%%%%%%%%%%%%%%%%%%%%%%%%%%%%%%%%%
%%%%%%%%%%%%%%%%%%%%%%%%%%%%%%%%%%%%%%%%%%%%%%%%%%%%%%%%%%%%%%%%%%%%%

%%%%%%%%%%%%%%%%%%%%%%%%%%%%%%%%%%%%%%%%%%%%%%%%%%%%%%%%%%%%%%%%%%%%%
%%%%%%%%%%%%%%%%%%%%%%%%%%%%%%%%%%%%%%%%%%%%%%%%%%%%%%%%%%%%%%%%%%%%%

\section{Krasovskii-passivity-based Stabilizing Control for Sampled Discrete-time Systems}\label{sec:stab}
In this section, we develop stabilizing control techniques for sampled discrete-time systems based on Krasovskii passivity. Then, we apply the proposed result to a linear PHS as an illustration.

As in the continuous-time case, stabilizing control design is studied under the following detectability property for the extended system~\eqref{eq:sys_ex}.
\smallskip
\begin{secdefn}
Under Assumption~\ref{asm:KP}, the extended system~\eqref{eq:sys_ex} is said to be \emph{detectable} at~$(x^*, u^*) \in \cE$ if 
\begin{align}
v_k = 
z_k=0, \, \forall k \in \bZ_+
 \implies  \lim_{k \to \infty} (x_k, u_k) = (x^*, u^*).
\end{align}
\red
\end{secdefn}

\subsection{Main Results}
Under Assumption~\ref{asm:KP}, we consider the following discrete-time feedback controller for the extended system~\eqref{eq:sys_ex}:
\begin{align}\label{eq:con_stab}
\Delta_\delta  u_k =   \alpha_\delta (\sigma_\delta u_k, z_k),
\end{align}
where $\alpha_\delta : \bR^m \times  \bR^m \to \bR^m$ is continuous at each $\delta$ such that $\alpha_\delta (u^*, 0) = 0$ {for $(x^*, u^*) \in \cE$}; recall $\sigma_\delta u^* = u^*$ from~\eqref{eq:const}. 

Based on the continuous-time case, we suppose that the controller~\eqref{eq:con_stab} possesses a kind of strict Krasovskii passivity. 
\smallskip
\begin{secasm}\label{asm:stab}
For the controller~\eqref{eq:con_stab}, there exist a continuous function $S_u: \bR^m \to \bR_+$ and $c > 0$ such that $S_u(u^*)=0$ and
\begin{align}\label{eq:con_KP}
\Delta_\delta S_u(u_k) \le - c |\Delta_\delta u_k|^2 - (\Delta_\delta u_k)^\top z_k,
\end{align}
for all $u_k, z_k \in \bR^m$. 
\red
\end{secasm}\smallskip

As the main result of this section, we extend Proposition~\ref{cprop:stab} for continuous-time systems to sampled discrete-time systems as follows.
\smallskip
\begin{secthm}\label{thm:stab}
Under Assumptions~\ref{asm:KP} and~\ref{asm:stab}, consider the closed-loop system consisting of a Krasovskii passive sampled discrete-time system~\eqref{eq:sys} at $(x^*, u^*) \in \cE$ and a discrete-time controller~\eqref{eq:con_stab}. If the closed-loop system is well defined, the following two statements hold:
\begin{itemize}
\item[(a)] 
If the closed-loop system is positively invariant on a compact set $\Omega \subset \bR^n \times \bR^m$ containing $(x^*, u^*)$,  then any trajectory starting from~$\Omega$ converges to the largest invariant set contained in
\begin{align*}
&\{ (x_k, u_k) \in \Omega : \hat z_k \equiv 0, \; W_K (\sigma_\delta x_k, u_k, \Delta_\delta \sigma_\delta x_k ) \equiv 0 \}\\
&\qquad \hat z_k := \alpha_\delta (\sigma_\delta u_k,  q_\delta ( \sigma_\delta x_k, u_k, \Delta_\delta \sigma_\delta x_k)).
\end{align*}

\item[(b)] If the extended system~\eqref{eq:sys_ex} with the {output~$(\hat z_k, W_K)$} is detectable at~$(x^*,u^*)$, then $(x^*, u^*)$ is asymptotically stable.
\end{itemize}
\end{secthm}\smallskip

\begin{proof}
The proof is in Appendix~\ref{app:stab}.
\end{proof}\smallskip

In the controller equation~\eqref{eq:con_stab},  we design $\Delta_\delta u_k$. This can be understood as an integral-type control design. The assumptions $\alpha_\delta (u^*, 0) = 0$ and $S_u(u^*) = 0$ {correspond to  stabilizing   $u_k$ to $u^*$}.  As mentioned in the continuous-time case, for some applications such as the boost power converter, specifying $u^*$ is enough to stabilize an equilibrium point $(x^*, u^*) \in \cE$ even if $x^*$ depends on unknown parameters. The same robustness is inherited in the sampled discrete-time case as illustrated through the network of boost converters considered in Section~\ref{sec:ex_bst} below.

%%%%%%%%%%%%%%%%%%%%%%%%%%%%%%%%%%%%%%%%%%%%%%%%%%%%%%%%%%%%%%%%%%%%%
%%%%%%%%%%%%%%%%%%%%%%%%%%%%%%%%%%%%%%%%%%%%%%%%%%%%%%%%%%%%%%%%%%%%%

\subsection{Illustration by Linear Port-Hamiltonian Systems}\label{ex2:LPH}
In this subsection, we apply Theorem~\ref{thm:stab} to linear PHSs in order to obtain a deeper understanding.
As a stabilizing controller, we use the implicit midpoint discretization of the continuous-time stabilizing controller~\eqref{ceq:con_stab}. First, we confirm that this satisfies Assumption~\ref{asm:stab} as follows.
\smallskip
\begin{secprop}\label{prop:con_stab_IMM}
The implicit midpoint discretization of the continuous-time controller~\eqref{ceq:con_stab}:
\begin{align}\label{eq:con_stab_IMM}
K_1 \Delta_\delta u_k =  K_2 (u^* - \sigma_\delta u_k) - z_k,
\end{align}
with $\sigma_\delta$ in~\eqref{eq2:IMM} satisfies Assumption~\ref{asm:stab}.
\end{secprop}

\begin{proof}
The proof is in Appendix~\ref{app:con_stab_IMM}.
\end{proof}\smallskip

Next, we show that~\eqref{eq:con_stab_IMM} is in fact a stabilizing controller if $u^*$ has a one-to-one correspondence with $(x^*, u^*) \in \cE$.
\smallskip
\begin{secthm}\label{lPHthm:stab}
The closed-loop system consisting of a sampled discrete-time linear PHS~\eqref{lPHeq:sys} and a sampled discrete-time controller~\eqref{eq:con_stab_IMM} is exponentially stable at $(x^*, u^*) \in \cE$ if $H, R \succ 0$, $B$ is of full column rank, and $I_n/\delta - (J - R) H/2$ and
\begin{align*}
A_s : = 
\begin{bmatrix}
I_n/\delta - (J - R)H/2 & - B\\
B^\top H/(2\delta ) & K_1/\delta + K_2/2
\end{bmatrix}
\end{align*}
are invertible. 
\end{secthm}\smallskip

\begin{proof}
The proof is in Appendix~\ref{lPHapp:stab}.
\end{proof}\smallskip 

For linear PHSs, it is standard to assume $H \succ 0$. The requirement $R \succ 0$ can be relaxed since $R$ can be shifted to $R + B K B^\top$ by applying the output feedback $u_k = - K y_k + \hat u_k$. Finally, given $J, R, H$, there always exists a sampling period $\delta > 0$ such that $I_n/\delta  - (J - R) H/2$ and $A_s$ are invertible.

%%%%%%%%%%%%%%%%%%%%%%%%%%%%%%%%%%%%%%%%%%%%%%%%%%%%%%%%%%%%%%%%%%%%%
%%%%%%%%%%%%%%%%%%%%%%%%%%%%%%%%%%%%%%%%%%%%%%%%%%%%%%%%%%%%%%%%%%%%%

\begin{figure}[t]
\begin{center}
\begin{tikzpicture}[scale=1, transform shape]

\node [text width=7cm, text height=1.5cm, fill=white, dashed, draw=gray, rectangle, rounded corners, text centered] at (3.15,0.2){ };
\draw (0.75,1.3) node
{Sampled model};

\bXInput{A}

\bXStyleBloc{rounded corners,fill=blue!20,text=blue}
\bXBloc[0]{zoh}{Holder}{A}

\bXStyleBloc{rounded corners,fill=black!10,text=blue}
\bXBloc[4]{plant}{Plant}{zoh}
\bXLink[$u(t)$]{zoh}{plant}

\bXStyleBloc{rounded corners,fill=blue!20,text=blue}
\bXBloc[4]{sampler}{Sampler}{plant}
\bXLink[$y(t)$]{plant}{sampler}

\bXOutput[2]{Z}{sampler}
\bXBranchy[5]{Z}{shift}
\bXLineStyle{-}
\bXLinkxy{sampler}{shift}

\bXStyleBloc{rounded corners,fill=blue!20,text=blue}
\bXBloc[-13.45]{controller}{Controller}{shift}
\bXOutput[-13.2]{shiftC}{controller}
\bXDefaultLineStyle
\bXLinkyx{Z}{controller}
\bXLink[$y(k\delta)$]{shift}{controller}
\bXLineStyle{-}
\bXLink[$u(k\delta)$]{controller}{shiftC}
\bXDefaultLineStyle
\bXLinkyx{shiftC}{zoh}

\bXLineStyle{blue, dotted}
\bXBranchy[-4]{plant}{parameters}
\bXLink{parameters}{plant}
{\large{
\bXLabelStyle{text=blue}
\bXLinkName[0.5]{parameters}{$d(t)$}}}

\end{tikzpicture}
\caption{Sampled-data scheme: the input signal of the plant is generated by a
holding device and the output is sampled at the same uniform distributed instants.}
\label{fig:scheme}
\end{center}
\end{figure}
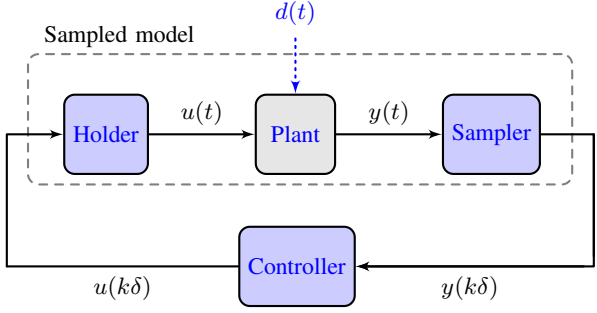

%%%%%%%%%%%%%%%%%%%%%%%%%%%%%%%%%%%%%%%%%%%%%%%%%%%%%%%%%%%%%%%%%%%%%
%%%%%%%%%%%%%%%%%%%%%%%%%%%%%%%%%%%%%%%%%%%%%%%%%%%%%%%%%%%%%%%%%%%%%

% \begin{figure}[t!]
%     \centering
%     \includegraphics[width=\linewidth]{Block_scheme2.pdf}
%     \caption{Sampled-data scheme: the input signal of the plant is generated by a
% holding device and the output is sampled at the same uniform distributed instants.}
%     \label{fig:scheme}
% \end{figure}

\section{Krasovskii-passivity-based Output Consensus Control for Sampled Discrete-time Systems}\label{sec:oc}
In this section, we propose strict Krasovskii-passivity-based control techniques to achieve output consensus for sampled discrete-time systems under unknown constant disturbance. We first provide an output consensus controller and then apply it to a linear PHS as an illustration. Throughout this section, we do not assume the existence of an equilibrium point.

\subsection{Main Result}
In this section, we consider nonlinear systems under unknown constant disturbance:
\begin{align}\label{eq:sys_d}
\left\{\begin{array}{r@{}l}
\Delta_\delta x_k &{}= g_\delta (\sigma_\delta x_k, u_k, d) \\[0.5ex]
y_k  &{}= h_\delta (\sigma_\delta x_k, d)
\end{array}\right.
\end{align}
with the output $y_k \in \bR^m$, $k \in \bZ_+$ and unknown constant disturbance $d \in \bR^r$, where $g_\delta : \bR^n \times \bR^m \times \bR^r \to \bR^n$ and $h_\delta:\bR^n \times \bR^r \to  \bR^m$ are continuous at each $\delta$. In the continuous-time case, we have supposed that $\partial g/\partial u$ in~\eqref{ceq:sys_d} is of full column rank to guarantee that $\dot x=0$ (and $\ddot x = 0$) implies $\dot u = 0$. Similarly, suppose that
\begin{align}\label{eq:sys_d_fc}
\lim_{k \to \infty} \Delta_\delta \sigma_\delta x_k = 0
\quad \implies \quad
\lim_{k \to \infty} \Delta_\delta u_k = 0.
\end{align}
We also assume that the system is well-defined.

To achieve output consensus, we consider the following output feedback controller:
\begin{align}\label{eq:con_oc}
\left\{\begin{array}{r@{}l}
\Delta_\delta \xi_k  &{}=   \beta_\delta ( \sigma_\delta \xi_k, \sigma_\delta y_k)\\[0.5ex]
u_k  &{}= E \xi_k,
\end{array}\right.
\end{align}
with the controller state $\xi_k \in \bR^p$, $k \in \bZ_+$, where $\beta_\delta : \bR^p \times \bR^m \to \bR^p$ is continuous at each $\delta$, and $E \in \bR^{m \times p}$ is such that ${\rm rank} \; E^\top = m-1$ and $E^\top \1_m = 0$.

Based on the continuous-time case, we require the following properties including Krasovskii passivity for the controller. 
\smallskip
\begin{secasm}\label{asm:oc}
For the controller~\eqref{eq:con_oc}, there exists a continuous function $S_y: \bR^m \to \bR_+$ such that 
\begin{align}\label{eq1:con_oc_KP}
\Delta_\delta S_y(y_k) \le  - (\Delta_\delta u_k)^\top \Delta_\delta y_k,
\end{align}
for all $\xi_k \in \bR^p$ and $y_k \in \bR^m$. Also, it follows that
\begin{align}\label{eq2:con_oc_KP}
\lim_{k \to \infty} \Delta_\delta u_k = 0
\quad \implies \quad
\lim_{k \to \infty} E^\top \sigma_\delta y_k = 0
\end{align}
for all $\xi_k \in \bR^p$ and $y_k \in \bR^m$.
\red
\end{secasm}\smallskip

As the main result of this section, we show that a sampled discrete-time controller~\eqref{eq:con_oc} achieves output consensus for a strictly Krasovskii passive sampled discrete-time system.
\smallskip
\begin{secthm}\label{thm:oc}
Given $d \in \bR^r$, consider the closed-loop system consisting of a strictly Krasovskii passive sampled discrete-time system~\eqref{eq:sys} and a discrete-time output feedback controller~\eqref{eq:con_oc} satisfying~\eqref{eq:sys_d_fc} and Assumption~\ref{asm:oc}, respectively.
If the closed-loop system is well defined and is positively invariant on a compact set $\Omega \subset \bR^n \times \bR^p$,  then for each $(x_0, \xi_0) \in \Omega$, there exists some $\{\alpha_k\}_{k \in \bZ_+} \in \bR^{\bZ_+}$ such that
\begin{align}\label{eq:oc}
\lim_{k \to \infty} (y_k - \alpha_k \1_m)= 0
\end{align}
holds.
\end{secthm}\smallskip

\begin{proof}
The proof is in Appendix~\ref{app:oc}.
\end{proof}\smallskip

%%%%%%%%%%%%%%%%%%%%%%%%%%%%%%%%%%%%%%%%%%%%%%%%%%%%%%%%%%%%%%%%%%%%%
%%%%%%%%%%%%%%%%%%%%%%%%%%%%%%%%%%%%%%%%%%%%%%%%%%%%%%%%%%%%%%%%%%%%%

\subsection{Illustration by Linear Port-Hamiltonian Systems}\label{ex3:LPH}
In this subsection, we design an output feedback controller to achieve output consensus for a linear PHS under unknown constant disturbance:
\begin{align*}
\left\{\begin{array}{r@{}l}
\dot x &{}=   (J - R) H x  + B u + d \\[0.5ex]
y &{}=  B^\top H x.
\end{array}\right.
\end{align*}
By applying the implicit midpoint method~\eqref{eq:IMM}, we get the following sampled discrete-time model:
\begin{align}\label{lPHeq:sys_d}
\left\{\begin{array}{r@{}l}
\Delta_\delta x_k &{}= (J - R) H \sigma_\delta x_k + B u_k + d\\[0.5ex]
y_k &{}=  B^\top H \sigma_\delta x_k.
\end{array}\right.
\end{align}
This is a sampled discrete-time linear PHS~\eqref{lPHeq:sys} with unknown constant disturbance $d \in \bR^n$. 

Here, we consider achieving a more general consensus, weighted output consensus:
\begin{align}\label{eq:woc}
\lim_{k \to \infty} M y_k = \alpha \1_m,
\end{align}
where $M \in \bR^{m \times m}$ denotes the weight.

As a weighted output consensus controller, we use the implicit midpoint discretization of a continuous-time output feedback controller proposed by~\cite{KCF:22}:
\begin{align}\label{eq1:con_oc_IMM}
\left\{\begin{array}{r@{}l}
\Delta_\delta \xi_k &{}= - E^\top M \sigma_\delta y_k\\[0.5ex]
\Delta_{\delta} \rho_k  &{}= \sigma_\delta y_k - \sigma_\delta \rho_k\\[0.5ex]
u_k &{}=M^\top E \xi_k - K (y_k - \rho_k)
\end{array}\right.
\end{align}
with the controller state $(\xi_k, \rho_k) \in \bR^p \times \bR^m$ and $\sigma_\delta$ in~\eqref{eq2:IMM}, where $0 \preceq K \in \bR^{m \times m}$ is a tuning parameter. This controller contains the implicit midpoint discretization of the continuous-time output consensus controller~\eqref{ceq:con_oc} when $M=I_m$ and $K = 0$, which satisfies Assumption~\ref{asm:oc}.

Computing $\Delta_\delta u_k$, this controller can be rewritten as 
\begin{align}\label{eq2:con_oc_IMM}
\left\{\begin{array}{r@{}l}
\Delta_\delta  u_k &{}= - M^\top E E^\top M \sigma_\delta y_k - K ( \Delta_\delta y_k - \Delta_\delta \rho_k)\\[0.5ex]
\Delta_{\delta} \rho_k  &{}= \sigma_\delta y_k - \sigma_\delta \rho_k.
\end{array}\right.
\end{align} 
Using this representation, we can now show weighted output consensus for a linear PHS~\eqref{lPHeq:sys_d}.
\smallskip
\begin{secthm}\label{lPHthm:oc}
Given $d \in \bR^n$, the closed-loop system consisting of a sampled discrete-time linear PHS~\eqref{lPHeq:sys_d} and a sampled discrete-time output feedback controller~\eqref{eq2:con_oc_IMM} achieves weighted output consensus~\eqref{eq:woc} for each $(x_0, u_0, \rho_0) \in \bR^n \times \bR^m \times \bR^m$ if $H, R \succ 0$, $K \succ 0$, $B$ is of full column rank, and $M$, $I_n/\delta - (J-R)H/2$, and
\begin{align*}
&A_c : = 
\begin{bmatrix}
I_n/\delta - (J - R)H/2 & - B & 0\\
A_{c, 12}& I_m/\delta & K/\delta \\
- B^\top H/4 & 0 & I_m/\delta + I_m/2
\end{bmatrix}\\[0.5ex]
&\quad 
A_{c, 12} = M^\top E E^\top M B^\top H/4 + K B^\top H/ (2\delta )
\end{align*}
are invertible. Moreover, the consensus value $\alpha$ is the same for all $(x_0, u_0, \rho_0) \in \bR^n \times \bR^m \times \bR^m$, and the closed-loop system is exponentially stable at some $(x^*, u^*, \rho^*) \in \bR^n \times \bR^m \times \bR^m$.
\end{secthm}\smallskip

\begin{proof}
The proof is in Appendix~\ref{lPHapp:oc}.
\end{proof}\smallskip 

The essence of the proof of Theorem~\ref{lPHthm:oc} is based on that of Theorem~\ref{thm:oc}. Especially for $K = 0$, weighted output consensus follows from Theorem~\ref{thm:oc}. Exponential stability can be shown by slightly modifying the proof of Theorem~\ref{lPHthm:oc}; for more details, see our preliminary version~\cite[Theorem 8]{KMC:22} focusing on this case. Note that, when $M=I_m$ we conclude that the implicit midpoint discretization of a continuous-time controller~\eqref{ceq:con_oc} achieves output consensus for each $(x_0, u_0) \in \bR^n \times \bR^m$, and the closed-loop system is exponentially stable. If $K \succeq 0$, output consensus can be shown on a compact invariant set of $(x_k, u_k, \rho_k)$, $k \in \bZ_+$. 

%\MC{The way it's written it gives to me the impression, it works only for $M=I$!}

%%%%%%%%%%%%%%%%%%%%%%%%%%%%%%%%%%%%%%%%%%%%%%%%%%%%%%%%%%%%%%%%%%%%%
%%%%%%%%%%%%%%%%%%%%%%%%%%%%%%%%%%%%%%%%%%%%%%%%%%%%%%%%%%%%%%%%%%%%%

%%%%%%%%%%%%%%%%%%%%%%%%%%%%%%%%%%%%%%%%%%%%%%%%%%%%%%%%%%%%%%%%%%%%%
%%%%%%%%%%%%%%%%%%%%%%%%%%%%%%%%%%%%%%%%%%%%%%%%%%%%%%%%%%%%%%%%%%%%%

\section{Examples}\label{sec:ex}
In this section, we apply the proposed sampled control techniques to Krasovskii passive nonlinear power networks.
First, we consider voltage regulation (\emph{i.e.}, stabilization) for a DC network of boost converters and current sharing
(\emph{i.e.}, output consensus control) for a DC network of buck converters with nonlinear loads.

\subsection{Stabilization of Boost Converters}\label{sec:ex_bst}

%%%%%%%%%%%%%%%%%%%%%%%%%%%%%%%%%%%%%%%%%%%%%%%%%%%%%%%%%%%%%%%%%%%%%
%%%%%%%%%%%%%%%%%%%%%%%%%%%%%%%%%%%%%%%%%%%%%%%%%%%%%%%%%%%%%%%%%%%%%

\begin{figure}[t]
\begin{center}
\ctikzset{bipoles/length=0.7cm}
\begin{circuitikz}[scale=1, transform shape]
\ctikzset{current/distance=1}
\draw
% transformators i and j
node[] (Ti) at (0,0) {}
node[] (Tj) at ($(5.4,0)$) {}
% Boost i
node[ocirc] (Aibattery) at ([xshift=-4.5cm,yshift=1.1cm]Ti) {}
node[ocirc] (Bibattery) at ([xshift=-4.5cm,yshift=-1.1cm]Ti) {}
(Aibattery) to [battery, l_=\scriptsize{$V^*_{si}$},*-*] (Bibattery) {}
node [rectangle,draw, minimum width=1cm,minimum height=2.4cm] (boosti) at ($0.5*(Aibattery)+0.5*(Bibattery)+(2.6,0)$) {\scriptsize{Switch$_i$}}
% filter i
node[ocirc] (Ai) at ($(Aibattery)+(1.8,0)$) {}
node[ocirc] (Bi) at ($(Bibattery)+(1.8,0)$) {}
(Ai) to [short] ([xshift=0.3cm]Ai)
(Bi) to [short] ([xshift=0.3cm]Bi)
(Aibattery) 
to [short,i=\scriptsize{$I_{si}$}]($(Aibattery)+(0.6,0)$){}
to [L, l=\scriptsize{$L_{si}$}]($(Aibattery)+(1.9,0)$){}
(Bibattery) to [short] ([xshift=2.1cm]Bibattery)
node[ocirc] (AAi) at ($(Ai)+(1.6,0)$) {}
node[ocirc] (BBi) at ($(Bi)+(1.6,0)$) {}
(AAi) to [short] ([xshift=-0.3cm]AAi)
(BBi) to [short] ([xshift=-0.3cm]BBi)
(AAi) to [short,i=\scriptsize{$(1-u_i)I_{si}$}]($(AAi)+(0.3,0)$){}
(AAi) to [short] ([xshift=0.7cm]AAi)
(BBi) to [short] ([xshift=0.7cm]BBi);
\begin{scope}[shorten >= 10pt,shorten <= 10pt,]
\draw[<-] (Ai) -- node[left] {\scriptsize{$(1-u_{i})V_i$}} (Bi);
\end{scope};
\draw
% PCC-i
($(Ti)+(0.0,1.1)$) node[anchor=north]{\scriptsize{$V_{i}$}}
%($(Ti)+(0.4,1.1)$) node[anchor=south]{\scriptsize{$PCC_i$}}
($(Ti)+(0.0,1.1)$) node[ocirc](PCCi){}
($(Ti)+(0.65,1.1)$)--($(Ti)+(0.65,0.8)$) to [short,i_>=\scriptsize{$I_{li}(V_i)$}]($(Ti)+(0.65,0.5)$)
to [I]($(Ti)+(0.65,-1.1)$)
($(Ti)+(-0.4,1.1)$) to [C, l_=\scriptsize{$C_{i}$}] ($(Ti)+(-0.4,-1.1)$)
% line
($(Ti)+(-0.4,1.1)$) to [short] ($(Ti)+(1.0,1.1)$)
($(Ti)+(1.8,1.1)$) to [short,i_=\scriptsize{$I_{k}$}] ($(Ti)+(2.0,1.1)$)--($(Ti)+(2.0,1.1)$)
($(Ti)+(0.6,1.1)$)--($(Ti)+(0.8,1.1)$) to [R, l=\scriptsize{$R_{k}$}] ($(Ti)+(2.0,1.1)$) {}
to [L, l={\scriptsize{$L_{k}$}}, color=black]($(Tj)+(-2.5,1.1)$){}
($(Ti)+(-0.4,-1.1)$) to [short] ($(Tj)+(-2.5,-1.1)$);
\draw
    			node [rectangle,draw,minimum width=6.15cm,minimum height=3.4cm,dashed,color=gray,label=\textbf{Node $i$},densely dashed, rounded corners] (DGUi) at ($0.5*(Aibattery)+0.5*(Bibattery)+(2.35,0.2)$) {}
    			node [rectangle,draw,minimum width=1.8cm,minimum height=3.4cm,dashed,color=gray,label=\textbf{Line $k$},densely dashed, rounded corners] (DGUi) at ($0.5*(Aibattery)+0.5*(Bibattery)+(6.4,0.2)$) {};
\end{circuitikz}
\caption{Electrical scheme of node $i$ and line $k$, where $I_{li}(V_i) = G_{li}^\ast V_i + I_{li}^\ast$.}
\label{fig:microgrid1}
\end{center}
\end{figure}
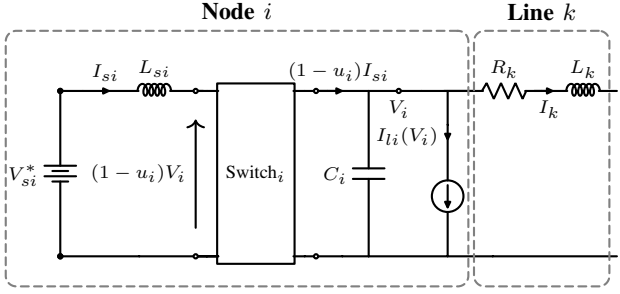

Consider a DC microgrid of boost converters in the presence of constant impedance and  current loads  with $\nu$ nodes and $\mu$ edges (see \emph{e.g.},~\cite{CLK:19}):
\begin{align}\label{bsteq:sys}
\left\{\begin{array}{r@{}l}
L_s \dot I_s &{}= - (I_\nu - {\rm diag}\{u\}) V + V^*_s\\[0.5ex]
C \dot V &{}= (I_\nu - {\rm diag}\{u\} )I_s - G^*_l V - I^*_l + D I\\[0.5ex]
L \dot I &{}= -D^\top V - RI,
\end{array}\right.
\end{align}
with the state $I_s, V \in \bR^\nu$ and $I \in \bR^\mu$, and input $u \in [0,1]^\nu$, representing the duty ratio of each boost converter. Moreover, $V^*_s, I^*_l \in \bR^\nu$, while $L_s, C, G^*_l \in \bR^{\nu \times \nu}$ and $L, R \in \bR^{\mu \times \mu}$ are diagonal and positive definite matrices; see Fig.~\ref{fig:microgrid1} and Table~\ref{tab:symbols} for the meaning of each used symbol. The incidence matrix $D \in \bR^{\nu \times \mu}$ describes the network topology.
%%%%%%%%%%%%%%%%%%%%%%%%%%%%%%%%%%%%%%%%%%%%%%%%%%%%%%%%%%%%%%%%%%%%%
%%%%%%%%%%%%%%%%%%%%%%%%%%%%%%%%%%%%%%%%%%%%%%%%%%%%%%%%%%%%%%%%%%%%%

\begin{table}
\begin{center}
\caption{Description of the used symbols}\label{tab:symbols}
\begin{tabular}{clcl}
\toprule
&{\bf State variables} & &{\bf Loads}\\
\midrule
$I_s$						& Generated current &$G^*_l$						& Load conductance\\
$V$						& Load voltage &$I^*_l$						& Load current\\
$I$ 						& Line current \\
\midrule
&{\bf Inputs} & &{\bf Filters and lines}\\
\midrule
$u$						& Control input &$L_s, C$						& Filter inductance, capacitance\\
$V^*_s$					& Voltage source &$R, L$						& Line resistance, inductance\\
\bottomrule
\end{tabular}
\end{center}
\end{table}
%%%%%%%%%%%%%%%%%%%%%%%%%%%%%%%%%%%%%%%%%%%%%%%%%%%%%%%%%%%%%%%%%%%%%
%%%%%%%%%%%%%%%%%%%%%%%%%%%%%%%%%%%%%%%%%%%%%%%%%%%%%%%%%%%%%%%%%%%%%
%%%%%%%%%%%%%%%%%%%%%%%%%%%%%%%%%%%%%%%%%%%%%%%%%%%%%%%%%%%%%%%%%%%%%
%%%%%%%%%%%%%%%%%%%%%%%%%%%%%%%%%%%%%%%%%%%%%%%%%%%%%%%%%%%%%%%%%%%%%
\begin{table}[t]
\begin{center}
\label{tab:symbols2}
\caption{Description of the used symbols}
\begin{tabular}{clcl}
\toprule
$\varphi$ & Flux (node) & $G^*_L$ & Load conductance\\
$q$ & Charge & $I^*_L$ & Load current\\
$\varphi_t$ & Flux (power line)& $P^*_L$ & Load power\\
$u$ & Control input & $R, L, C$ & Filter parameters\\
$d$ & Disturbance & $R_t, L_t$ & Line parameters\\
\bottomrule
\end{tabular}
\end{center}
\end{table}
%%%%%%%%%%%%%%%%%%%%%%%%%%%%%%%%%%%%%%%%%%%%%%%%%%%%%%%%%%%%%%%%%%%%%
%%%%%%%%%%%%%%%%%%%%%%%%%%%%%%%%%%%%%%%%%%%%%%%%%%%%%%%%%%%%%%%%%%%%%
The control goal is voltage regulation, \emph{i.e.},
\begin{align}\label{bsteq:volreg}
\lim_{t \to \infty} V(t) = V^*,
\end{align}
where $V_i^* \geq V_{si}^*$ for all $i =1,\dots,\nu$.
We use the sampled discrete-time controller~\eqref{eq:con_stab_IMM} to solve the voltage regulation problem.
As a discretization method for the network of boost converters~\eqref{bsteq:sys}, we consider the implicit midpoint method, {\emph{i.e.}},
\begin{align}\label{bsteq:sys_IMM}
\left\{\begin{array}{r@{}l}
L_s \Delta_\delta I_{s,k} &{}= -  (I_\nu - {\rm diag}\{u_k\}) \sigma_\delta V_k + V^*_s\\[0.5ex]
C \Delta_\delta V_k &{}=  (I_\nu - {\rm diag}\{u_k\}) \sigma_\delta I_{s,k} - G^*_l \sigma_\delta V_k\\[0.5ex]
&\qquad - I^*_l + D \sigma_\delta I_k\\[0.5ex]
L \Delta_\delta I_k &{}= -D^\top \sigma_\delta V_k - R \sigma_\delta I_k,
\end{array}\right.
\end{align}
with $\sigma_\delta$ in~\eqref{eq2:IMM}. Note that, since $L_s, C \succ 0$ and $L \succ 0$, the system~\eqref{bsteq:sys_IMM} is well defined, \emph{i.e.}, the matrix
\begin{align*}
&\Pi :=\\
&\begin{bmatrix}
L_s/\delta & (I_\nu - {\rm diag}\{u_k\})/2  & 0\\
-(I_\nu - {\rm diag}\{u_k\})/2 & C/\delta + G_l^*/2& -D/2 \\
0 & D^\top/2 & L/\delta + R/2 
\end{bmatrix}
\nonumber
\end{align*}
is invertible for any $\delta > 0$.
Moreover, given $V^* \in \bR^\nu$ with $V_i^* \geq V_{si}^*$, $i =1,\dots,\nu$, the system~\eqref{bsteq:sys_IMM} has a unique equilibrium point $((I_s^*, V^*, I^*), u^*)$ satisfying
\begin{align}\label{bsteq:eq}
\left\{\begin{array}{r@{}l}
(I_\nu - {\rm diag}\{u^*\})  V^* &{}= V_s^* \\[0.5ex]
(I_\nu - {\rm diag}\{u^*\}) I_s^* &{}= G^*_l V^* + I^*_l  - D I^* \\[0.5ex]
I^* &{}= -R^{-1}D^\top V^*.
\end{array}\right.
\end{align}
Note also that $V_i^* \geq V_{si}^*$ guarantees $u_i^* \in [0, 1)$ for each $i =1,\dots,\nu$. 
For the sampled discrete-time controller~\eqref{eq:con_stab_IMM}, we only require information of $u^*$ which can be computed from the (known) values of the desired voltage $V^*$ and voltage source $V_s^*$ (see the first line of~\eqref{bsteq:eq}), without requiring any information on the other system parameters. In this sense, the sampled controller~\eqref{eq:con_stab_IMM} is robust with respect to parameter uncertainty. 

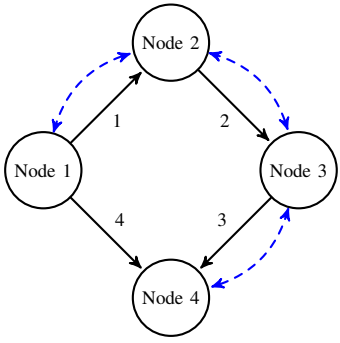
\begin{figure}[t]
\begin{center}
\begin{tikzpicture}[scale=0.8,transform shape,->,>=stealth',shorten >=1pt,auto,node distance=3cm,
                    semithick]
  \tikzstyle{every state}=[circle,thick,fill=white,draw=black,text=black]

  \node[state] (A)                    {Node 1};
  \node[state]         (B) [above right of=A] {Node 2};
  \node[state]         (D) [below right of=A] {Node 4};
  \node[state]         (C) [below right of=B] {Node 3};

  \path (A) edge   [below] node {~~~1} (B)
  		edge 	     node {4} (D)
           (B) edge      [below]        node {2~~~~} (C)
           (C) edge         [above left]     node {3} (D);

           \path[<->] (A) edge [bend left, dashed, blue]          node {} (B)
  		%edge	 [bend right, dashed, red]	     	node {} (D)
           (B) edge [bend left, dashed, blue]          	node {} (C)
           (C) edge [bend left, dashed, blue]          	node {} (D);
\end{tikzpicture}
\caption{Scheme of the considered microgrid composed of 4 nodes. The black solid arrows indicate the positive direction of the currents through the power lines. For the case of output consensus, the dashed blue lines represent the communication network.}
\label{fig:microgrid_example}
\end{center}
\end{figure}

Before applying the controller~\eqref{eq:con_stab_IMM}, we first show that the DC microgrid~\eqref{bsteq:sys_IMM} is Krasovskii passive with respect to the  storage function
\begin{align*}
S_K (\Delta_\delta x_k) : = \frac{1}{2} (|\Delta_\delta I_{s,k}|_{L_s}^2 + |\Delta_\delta V_k|_C^2 + |\Delta_\delta I_k|_L^2),
\end{align*}
which satisfies
\begin{align*}
\Delta_\delta S_K (\Delta_\delta x_k)
&= (\Delta_\delta \sigma_\delta I_{s,k}) L_s \Delta_\delta^2 I_{s,k}\\
&\qquad + (\Delta_\delta \sigma_\delta V_k) C \Delta_\delta^2 V_k
+ (\Delta_\delta \sigma_\delta I_k) L \Delta_\delta^2 I_k \\
&=   - |\Delta_\delta \sigma_\delta V_k|_{G^*_l}^2 - |\Delta_\delta \sigma_\delta I_k|_R^2\\
&\qquad +(\Delta_\delta \sigma_\delta I_{s,k})^\top \Delta_\delta (u_k \circ \sigma_\delta V_k ) \\
&\qquad - (\Delta_\delta \sigma_\delta V_k)^\top \Delta_\delta (u_k \circ \sigma_\delta I_{s,k}) \\
&= - (|\Delta_\delta \sigma_\delta I_k|_R^2 + |\Delta_\delta \sigma_\delta V_k|_{G^*_l}^2)\\
&\qquad + (\Delta_\delta u_k)^\top z_k,
\end{align*}
where $z_k = \Delta_\delta \sigma_\delta I_{s,k} \circ \sigma_\delta V_k - \sigma_\delta I_{s,k} \circ \Delta_\delta \sigma_\delta V_k$.
This implies Krasovskii passivity at $((I_s^*, V^*, I^*), u^*)$ with respect to the Krasovskii passive output $z_k$. However, since the Krasovskii passive output $z_k$ depends on $I_{s,k+1}, I_{s,k+2}, V_{k+1}$ and $V_{k+2}$, then the controller~\eqref{eq:con_stab_IMM} is not implementable in practice. Thus, in order to overcome this issue, one can use its (twice) backwards time-shifted version and approximate $u_k$ with $u_{k-1}$, \emph{i.e.},
\begin{subequations}\label{eq:shift_boost}
\begin{align}
    u_k \approx u_{k-1} =&~ \left(2K_1 + \delta K_2 \right)^{-1} \left( \left(2K_1 - \delta K_2 \right)u_{k-2} \right.\nonumber\\
    &\left. -2\delta z_{k-2} + 2\delta K_2 u^\ast \right),
\end{align}
where 
\begin{align}\label{eq:z_k_2}
z_{k-2} =&~\frac{1}{4\delta}\left((I_{s,k}-I_{s,k-2}) \circ (V_{k-1}+V_{k-2})\right.\nonumber\\
&\left. - (I_{s,k-1}+I_{s,k-2}) \circ (V_{k}-V_{k-2})\right).
\end{align}
\end{subequations}

Now, we are ready to show the effectiveness of the time-shifted version of~\eqref{eq:con_stab_IMM}, \emph{i.e.}, the controller~\eqref{eq:shift_boost} by performing voltage regulation (see the control objective~\eqref{bsteq:volreg}) in a DC microgrid of 4 boost converters interconnected as in Fig.~\ref{fig:microgrid_example}. According to~\cite{CUCUZZELLA2018161},  for each node $i=1, \cdots, 4$, we select $L_{si}=$ \SI{1.12}{\milli\henry}, $C_i=$ \SI{6.8}{\milli\farad}, and $V_{si}^\ast=$ \SI{280}{\volt}. The parameters of the lines are chosen as in \cite[Table III]{CTD:18}, while the values for the loads are taken from \cite[Table 2]{KCF:22}. For each node $i=1, \cdots, 4$, the controller gains are selected as $K_{1i} =$ \num{1e6} and $K_{2i}=$ \num{4e7}. In the simulation, at the time instant $t=$ \SI{1}{\second}, we consider a step-increase of the load $I_{l}$ equal to 50\%. Figure~\ref{fig:KPBC} shows that the voltages asymptotically converge to the desired value $V_i^\ast =$ \SI{380}{\volt}, and  the generated currents are asymptotically stable. Finally, we show in Fig.~\ref{fig:KPBC_u} the comparison between the sampled controller~\eqref{eq:shift_boost} (solid line) and its continuous-time version (dotted line).

\begin{figure}
\begin{center}
\includegraphics[width=1.05\columnwidth]{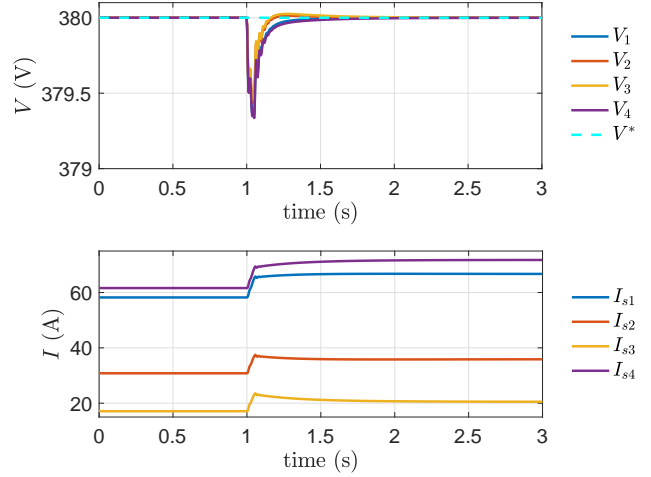}   
\caption{Stabilization of boost converters. {\bf (Top)}: Time evolution of the voltages together with the corresponding reference (dashed line). {\bf(Bottom)}: Time evolution of the generated currents.} 
\label{fig:KPBC}
\end{center}
\end{figure}

\begin{figure}
\begin{center}
\includegraphics[width=1.05\columnwidth]{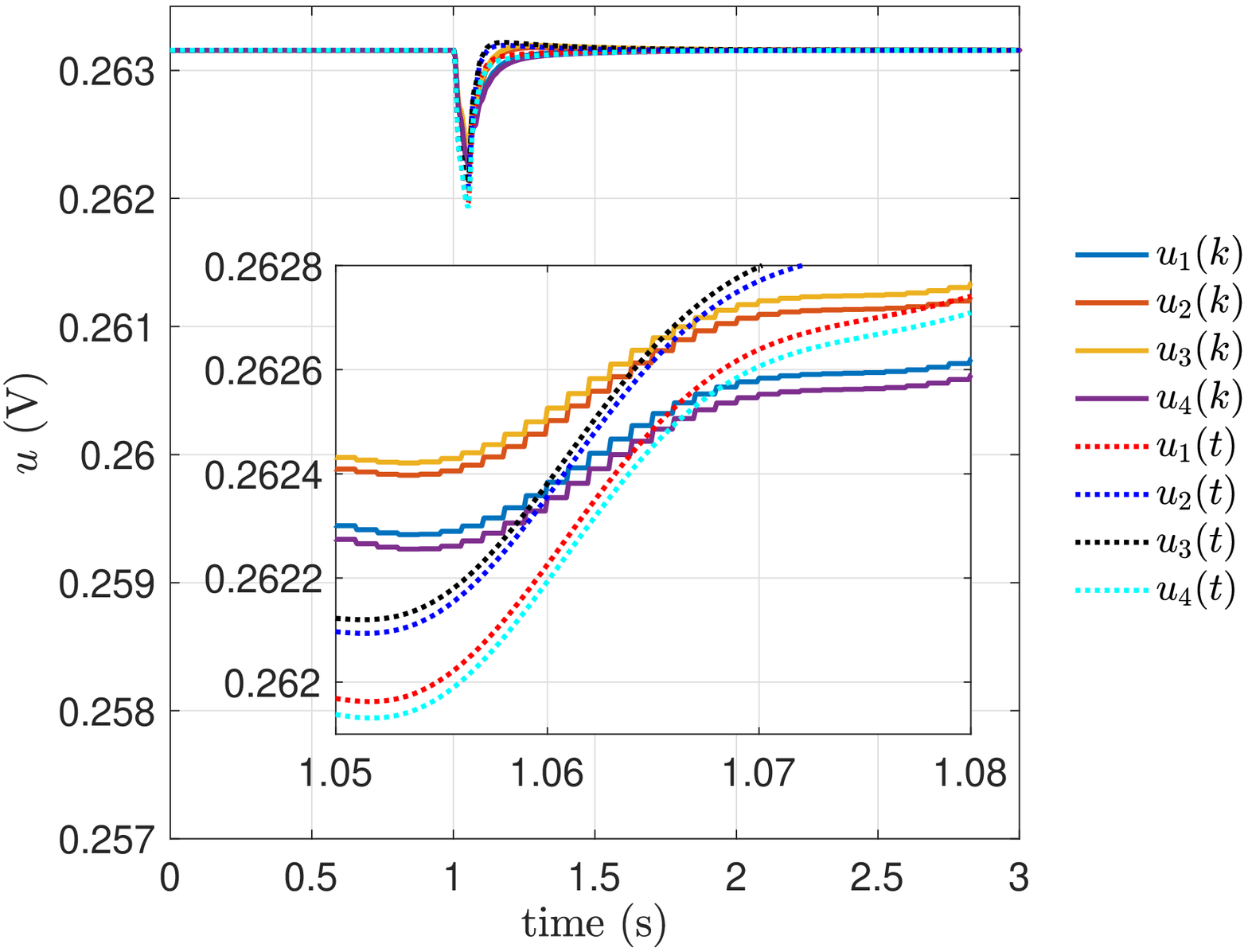}   
\caption{Continuous-time (dotted lines) and discrete-time (solid lines) control inputs.} 
\label{fig:KPBC_u}
\end{center}
\end{figure}

\begin{figure}[t]
    	\begin{center}
    		\begin{circuitikz}[scale=1,transform shape]
    			\ctikzset{current/distance=1}
    			\draw
    			% transformators i and j
    			node[] (Ti) at (0,0) {}
    			node[] (Tj) at ($(5.4,0)$) {}
    			% Buck i
    			node[] (Aibattery) at ([xshift=-4.5cm,yshift=0.9cm]Ti) {}
    			node[] (Bibattery) at ([xshift=-4.5cm,yshift=-0.9cm]Ti) {}
    			node[] (Ai) at ($(Aibattery)+(0,0.2)$) {}
    			node[] (Bi) at ($(Bibattery)+(0,-0.2)$) {}
    			($(Ai)+(-0.0005,0)$) to [R, l={$R_{i}$}] ($(Ai)+(1.7,0)$) {}
    			($(Ai)+(1.7,0)$) to [short,i_={$\dfrac{\varphi_i}{L_i}$}]($(Ai)+(1.701,0)$){}
    			($(Ai)+(1.701,0)$) to [L, l={$L_{i}$}] ($(Ai)+(3,0)$){}
    			to [short, l={}]($(Ti)+(0,1.1)$){}
    			($(Bi)+(-0.0005,0)$) to [short] ($(Ti)+(0,-1.1)$);
    			\draw
    			($(Ai)$) to []($(Aibattery)+(0,0)$)to [V_=$u_i$]($(Bi)$)
    			% PCC-i
    			($(Ti)+(-1.3,1.1)$) node[anchor=south]{{$\dfrac{q_i}{C_i}$}}
    			($(Ti)+(-1.3,1.1)$) node[ocirc](PCCi){}
    			($(Ti)+(-.3,1.1)$) to [short,i>={$I_{Li}(q_i)$}]($(Ti)+(-.3,0.5)$)to [I]($(Ti)+(-.3,-1.1)$)
    			($(Ti)+(-1.3,1.1)$) to [C, l_={$C_{i}$}] ($(Ti)+(-1.3,-1.1)$)
    			% line
    			($(Ti)+(2.,1.1)$) to [short,i_={$\dfrac{\varphi_{tk}}{L_{tk}}$}] ($(Ti)+(2.2,1.1)$)
    			($(Ti)+(0,1.1)$)--($(Ti)+(.6,1.1)$) to [R, l={$R_{tk}$}] 
    			($(Ti)+(2.5,1.1)$) {} to [L, l={{$L_{tk}$}}, color=black]($(Tj)+(-2.2,1.1)$){}
    			($(Tj)+(-2.2,1.1)$) to [short]  ($(Ti)+(3.4,1.1)$)
    			($(Ti)+(0,-1.1)$) to [short] ($(Ti)+(3.4,-1.1)$);
    			\draw
    			node [rectangle,draw,minimum width=6.2cm,minimum height=3.4cm,dashed,color=gray,label=\textbf{Node $i$},densely dashed, rounded corners] (DGUi) at ($0.5*(Aibattery)+0.5*(Bibattery)+(2.25,0.2)$) {}
    			node [rectangle,draw,minimum width=2.2cm,minimum height=3.4cm,dashed,color=gray,label=\textbf{Line $k$},densely dashed, rounded corners] (DGUi) at ($0.5*(Aibattery)+0.5*(Bibattery)+(6.65,0.2)$) {};
    		\end{circuitikz}
    		\caption{Electrical scheme of node $i$ and line $k$, where $I_{Li}(q_i):= G_{Li}^*\frac{q_i}{C_i} + I_{Li}^* + \frac{C_i}{q_i}P_{Li}^*$.}
    		\label{fig:networks}
    	\end{center}
    \end{figure}
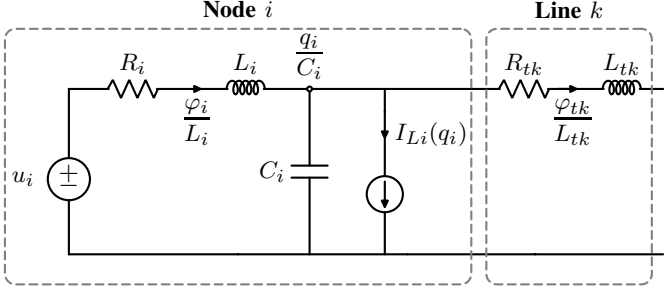

\subsection{Output consensus for Buck converters}

%%%%%%%%%%%%%%%%%%%%%%%%%%%%%%%%%%%%%%%%%%%%%%%%%%%%%%%%%%%%%%%%%%%%%
%%%%%%%%%%%%%%%%%%%%%%%%%%%%%%%%%%%%%%%%%%%%%%%%%%%%%%%%%%%%%%%%%%%%%

Consider a DC microgrid of buck converters in the presence of constant impedance,  current and power loads  with $\nu$ nodes and $\mu$ edges (see \emph{e.g.} \cite{FCS:23}):
\begin{align}\label{DCeq:sys}
&\left\{\begin{array}{r@{}l}
\dot x &{}= (\cJ - \cR) \nabla \cH (x) + \bar f(q) + g u  + d \\[0.5ex]
y &{}=g^\top \nabla \cH(x)
\end{array}\right.
\end{align}
where $u, y \in \bR^\nu$, and $d \in \bR^{2\nu + \mu}$ denote the input, the generated current, and an unknown constant disturbance, respectively, and the state vector consists of $\varphi, q \in \bR^\nu$ and $\varphi_t \in \bR^\mu$, \emph{i.e.},
$x := 
\begin{bmatrix}
\varphi^\top & q^\top & \varphi_t^\top
\end{bmatrix}^\top$. The Hamiltonian function is as follows $$\cH (x) := \frac{1}{2}\left( | \varphi |_{L^{-1}}^2 + | q |_{C^{-1}}^2 + | \varphi_t |_{L_t^{-1}}^2 \right),$$
and 
\begin{align}
&  \cJ :=
\begin{bmatrix}
0 & - I_\nu & 0 \\
I_\nu & 0 & D \\
0 & - D^\top & 0
\end{bmatrix}, \quad
\cR := 
\begin{bmatrix}
R & 0& 0 \\
0 & G_L^* & 0\\
0 & 0 & R_t
\end{bmatrix},
\nonumber\\
&  \bar f(q) := 
-
\begin{bmatrix}
0 \\ I_L^* + {\rm diag} \{ C^{-1} q\}^{-1} P_L^* \\ 0
\end{bmatrix}, \quad
g :=\begin{bmatrix}
I_\nu \\ 0 \\ 0
\end{bmatrix},
\nonumber
\end{align}
where $R, L, C \in \bR^{\nu \times \nu}$ and $R_t, L_t \in \bR^{\mu \times \mu}$ are diagonal and positive definite, and $I_L^*, P_L^* \in \bR^\nu$, and $G_L^* \in \bR^{\nu \times \nu}$; see Fig.~\ref{fig:networks} and Table~\ref{tab:symbols} for the meaning of each used symbol. The incidence matrix $D \in \bR^{\nu \times \mu}$ describes the network topology.

To improve the generation efficiency, it is generally desired in DC microgrids that the total current demand is shared among all the nodes  \cite{CTD:18}. This is called \emph{current sharing}, and in the considered application is equivalent to achieving output consensus, \emph{i.e.}, $\lim_{k \to \infty} (y_k - \alpha_k \1_m)= 0$.
We use the sampled discrete-time controller~\eqref{eq1:con_oc_IMM} to achieve current sharing.
As a discretization for the network of buck converters \eqref{DCeq:sys}, we consider the implicit midpoint method, {\emph{i.e.}},
\begin{align}\label{DCeq:sys_IMM}
\left\{\begin{array}{r@{}l}
\Delta_\delta x_k &{}= (\cJ - \cR) \nabla \cH( \sigma_\delta x_k) + \sigma_\delta \bar f (q_k) + g u_k  + d \\[0.5ex]
y_k &{}=g^\top \nabla \cH( \sigma_\delta x_k),
\end{array}\right.
\end{align}
with $\sigma_\delta$ in~\eqref{eq2:IMM}. 
To check the well definedness of the discretized model \eqref{DCeq:sys_IMM}, we consider the partial derivative of $\Delta_\delta x_k - (\cJ - \cR) \nabla \cH( \sigma_\delta x_k) - \sigma_\delta f (q_k) - g u_k  - d$ with respect to $x_{k+1}$, which is
\begin{align*}
\Pi (x_{k+1}) &:= \frac{1}{\delta} I_{2 \nu + \mu} - \frac{1}{2}(\cJ - \cR) {\rm diag} \{ L^{-1}, C^{-1}, L_t^{-1}\}\nonumber\\
&\qquad + \frac{1}{2} {\rm diag}\{0, {\rm diag}\{C P^*_L \} {\rm diag}\{q_{k+1} \}^{-2}, 0\}.
\end{align*}
Then, for any $x_{k+1}$ such that $\Pi (x_{k+1})$ is invertible, there exist open subsets $U_{x_{k+1},u_k}, V \subset \bR^\nu \times \bR^\nu \times \bR^\mu$ such that $U_{x_{k+1},u_k} \ni x_k \mapsto x_{k+1} \in V$ is well defined and analytic at each fixed $u_k \in \bR^\nu$. By the identity theorem, this mapping at each fixed $u_k \in \bR^\nu$ is well defined on the union of $U_{x_{k+1},u_k}$ such that $\Pi (x_{k+1})$ is invertible. By physics, $q_k$ is element-wise positive, and thus $\Pi(x_{k+1})$ is invertible for practically meaningful values of $x_k$ when the sampling period $\delta > 0$ is sufficiently small.

Next, we show that the sampled DC microgrid~\eqref{DCeq:sys_IMM} is strictly Krasovskii passive with respect to the  storage function
\begin{align*}
S_K (  \Delta_\delta x_k) : =  \frac{1}{2} |\Delta_\delta x_k|_{\nabla^2 \cH}^2,
\end{align*}
which satisfies
\begin{align*}
\Delta S_K (\Delta_\delta x_k) 
&= (\Delta_\delta \sigma_\delta x_k)^\top \nabla^2 \cH \Delta_\delta^2 x_k\\
&= - W_K(q_k, q_{k+2}, \Delta_\sigma \sigma_\delta x_k) + (\Delta_\delta u_k)^\top \Delta_\delta y_k,
\end{align*}
where
\begin{align*}
&W_K(q_k, q_{k+2}, \Delta_\sigma \sigma_\delta x_k) \\
&:=|\nabla^2 \cH \Delta_\delta \sigma_\delta x_k|_{\cR}^2 - (\Delta_\delta \sigma_\delta x_k)^\top \nabla^2 \cH \Delta_\delta \sigma_\delta \bar f(q_k)\\
&=|\nabla^2 \cH \Delta_\delta \sigma_\delta x_k|_{\cR}^2 - |\Delta_\delta \sigma_\delta q_k|_{{\rm diag}\{P_L^*\} {\rm diag}\{ q_k \circ q_{k+2}\}^{-1} }^2.
\end{align*}
Next,~\eqref{eq2:KP} holds if
\begin{align}\label{DCeq:W}
G_L^*  - C^2 {\rm diag}\{P_L^*\} {\rm diag}\{ q_k \circ q_{k+2}\}^{-1} \succ 0.
\end{align}
Although $q_{k+2}$ depends on $u_{k+1}$, by a slight modification of Theorem~\ref{thm:oc}, one can show that a sampled discrete-time controller~\eqref{eq1:con_oc_IMM} achieves output consensus if $\Pi(x_{k+1})$ is invertible, and~\eqref{DCeq:W} holds along the closed-loop trajectory staying in a compact set. We confirm this by numerical simulations.

We consider a DC microgrid of 4 buck converters interconnected as in Fig.~\ref{fig:microgrid_example}, where the dashed blue lines represent the communication network. The values of the parameters of each node and line are chosen as in \cite[Tables II, III]{CTD:18}, while those of the loads are taken from~\cite[Table 2]{KCF:22}. For each node, $i=1, \cdots, 4$, the controller parameters are selected as $M= 10I_{4}$, $K= 0.5 I_{4}$, and the matrix $E$ corresponds to the incidence matrix associated with the communication network in Fig.~\ref{fig:microgrid_example}. Moreover, we add the term $C^{-1}q^* + RL^{-1} \varphi_k$ to the designed control input $u_k$. This simply allows us to shift the system equilibrium such that the voltage average ($V_{\mathrm{av}}$) is equal to the voltage reference $V_i^\ast =$ \SI{380}{\volt}, $i=1, \cdots, 4$; see, \emph{e.g.}, \cite{TCC:19}. Note that, also in this case, in order to make the controller implementation possible in practice, we use its backward time-shifted version.

For the sake of notational  simplicity, let $V_i:={q_i}/{C_i}$ and $I_i:={\varphi_i}/{L_i}$ denote respectively the voltage and the generated current associated with node $i=1,\dots,4$. 
In the simulation, at the
time instant $t=$ \SI{1}{\second}, we consider a step-increase of the load $P_L^\ast$ as in Table~\cite[Table 2]{KCF:22}. Figure~\ref{fig:out_con} shows that both voltages and currents converge to a constant equilibrium. Specifically, the voltage average converges to the voltage reference, and output consensus (\emph{i.e.}, current sharing) is achieved. Finally, we show in Fig.~\ref{fig:out_con_u} the comparison between the sampled controller~\eqref{eq1:con_oc_IMM} (solid line) and its continuous-time version (dotted line).

\section{Conclusion}\label{sec:con}
In this paper, we have introduced the concept of Krasovskii passivity for sampled discrete-time nonlinear systems, inspired by the counterpart concept to continuous-time systems. The proposed concept has been investigated by relating it with widely known passivity concepts: incremental passivity and shifted passivity. In particular, we have established the following implications: incremental passivity $\implies$ Krasovskii passivity $\implies$ shifted passivity with respect to a suitable output function. Then, we have developed sampled-data control frameworks for stabilization and output consensus based on Krasovskii passivity. Their effectiveness has been illustrated by a network of boost converters and an islanded DC microgrid, respectively. As exemplified by these power systems and linear PHSs, suitable temporal discretization preserves Krasovskii passivity, \emph{i.e.}, suitable temporal discretizations of Krasovskii passive continuous-time systems are Krasovskii passive sampled discrete-time systems. Future work includes investigating if this is true for all Krasovskii passive continuous-time systems.

%%%%%%%%%%%%%%%%%%%%%%%%%%%%%%%%%%%%%%%%%%%%%%%%%%%%%%%%%%%%%%%%%%%%%
%%%%%%%%%%%%%%%%%%%%%%%%%%%%%%%%%%%%%%%%%%%%%%%%%%%%%%%%%%%%%%%%%%%%%

%%%%%%%%%%%%%%%%%%%%%%%%%%%%%%%%%%%%%%%%%%%%%
%%%%%%%%%%%%%%%%%%%%%%%%%%%%%%%%%%%%%%%%%%%%%
\begin{figure}[t!]
\begin{center}
\includegraphics[width=1.05\columnwidth]{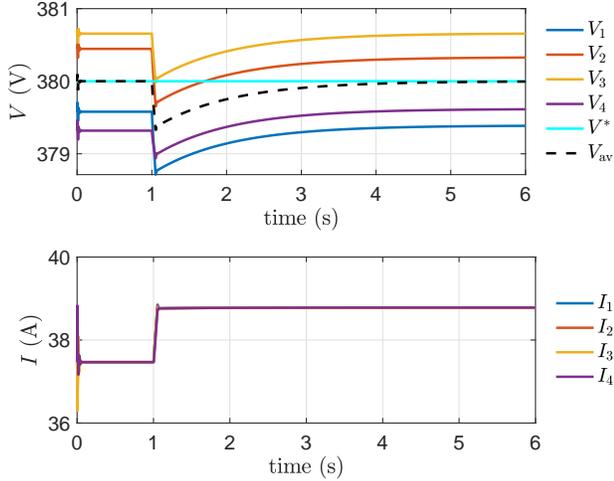}   
\caption{Output consensus for buck converters. {\bf (Top)}: Time evolution of the voltages, their average value (dashed line) and the corresponding reference (cyan line). {\bf (Bottom)}: Time evolution of the generated currents.} 
\label{fig:out_con}
\end{center}
\end{figure}
%%%%%%%%%%%%%%%%%%%%%%%%%%%%%%%%%%%%%%%%%%%%%%%%%%%%%%%%%%%%%%%%%%%%%
%%%%%%%%%%%%%%%%%%%%%%%%%%%%%%%%%%%%%%%%%%%%%%%%%%%%%%%%%%%%%%%%%%%%%
%%%%%%%%%%%%%%%%%%%%%%%%%%%%%%%%%%%%%%%%%%%%%%%%%%%%%%%%%%%%%%%%%%%%%
%%%%%%%%%%%%%%%%%%%%%%%%%%%%%%%%%%%%%%%%%%%%%%%%%%%%%%%%%%%%%%%%%%%%%

\appendices
\section{Proofs of Theorems in Section~\ref{sec:KP}}

\subsection{Proof of Theorem~\ref{thm:IP}}\label{app:IP}
Using a storage function for incremental passivity, define
\begin{align*}
\hat S_K(x_k, u_k) :&= \frac{1}{\delta^2} S_I( x_k, x_{k+1}) \\
&= \frac{1}{\delta^2} S_I( x_k, F_\delta (x_k, u_k)).
\end{align*}
Then, it holds that
\begin{align*}
\hat S_K(x^*, u^*)
&= \frac{1}{\delta^2} S_I( x^*, F_\delta (x^*, u^*))\\
&= \frac{1}{\delta^2} S_I( x^*, x^*) = 0.
\end{align*}
Next, it follows from~\eqref{eq:sys_ex_exp} and~\eqref{eq:IP} with $x_{k+2} = F_\delta (x_{k+1}, u_{k+1})$ that 
\begin{align*}
\Delta_\delta \hat S_K( x_k, u_k) 
&=  \frac{1}{\delta^2} \Delta_\delta S_I( x_k, x_{k+1})  \\
&\le \frac{(u_k - u_{k+1})^\top}{\delta} \frac{y_k - y_{k+1}}{\delta}\\
&= (\Delta_\delta u_k)^\top \Delta_\delta y_k = v_k^\top z_k.
\end{align*}
Thus, the system is Krasovskii passive.
\hfill\QED

%%%%%%%%%%%%%%%%%%%%%%%%%%%%%%%%%%%%%%%%%%%%%%%%%%%%%%%%%%%%%%%%%%%%%
%%%%%%%%%%%%%%%%%%%%%%%%%%%%%%%%%%%%%%%%%%%%%%%%%%%%%%%%%%%%%%%%%%%%%

\subsection{Proof of Theorem~\ref{thm:SP}}\label{app:SP}
With the explicit form~\eqref{eq:sys_ex_exp}, the dissipation inequality~\eqref{eq1:KP} for Krasovskii passivity can be rewrittens as
\begin{align*}
\frac{\hat S_K( F_\delta (x_k, u_k), u_{k+1}) - \hat S_K(x_k, u_k)}{\delta} \le \frac{(u_{k+1} -u_k)^\top}{\delta} z_k,
\end{align*}
where recall $v_k =(u_{k+1} -u_k)/\delta$ in~\eqref{eq:sys_ex_exp}.
Substituting $u_{k+1} = u_k = u^*$ into this yields
\begin{align}\label{pf1:SP}
\hat S_K( F_\delta (x_k, u^*), u^*) - \hat S_K(x_k, u^*) \le 0.
\end{align}
Utilizing this inequality, we show that a storage function for shifted passivity is $S_S(x_k) := \delta^2 \hat S_K(x_k, u^*)$.
First, it holds from~\eqref{eq:KP_storage_exp} that
\begin{align*}
S_S (x^*) &= \delta^2 \hat S_K(x^*, u^*)\\ 
&= \delta^2 S_K( x^*, u^*, 0)=0.
\end{align*}
Next, it follows from~\eqref{pf1:SP} and the fundamental theorem of calculus that 
\begin{align*}
\Delta_\delta S_S ( x_k) 
&=  \delta^2 \Delta_\delta \hat S_K( x_k, u^*)  \\
&= \delta (\hat S_K( F_\delta (x_k, u_k), u^*) - \hat S_K(x_k, u^*) ) \\
&\le \delta (\hat S_K( F_\delta (x_k, u_k), u^*) - \hat S_K( F_\delta (x_k, u^*), u^*))  \\
&= \delta \int_0^1 \frac{\partial \hat S_K( F_\delta (x_k, s u_k + (1-s) u^*), u^*)}{\partial s} ds.
\end{align*}
By the chain rule and the definition~\eqref{eq:SP_out} of $y_k$, we have
\begin{align*}
&\delta \!\int_0^1 \frac{\partial \hat S_K( F_\delta (x_k, s u_k + (1-s) u^*), u^*)}{\partial s} ds=(u_k - u^*)^\top y_k.
\end{align*}
Note that $y_k$ in~\eqref{eq:SP_out} is zero when $u_k = u^*$, \emph{i.e.}, $y^* =0$. Thus, the system is shifted passive. \hfill\hfill\QED

%%%%%%%%%%%%%%%%%%%%%%%%%%%%%%%%%%%%%%%%%%%%%%%%%%%%%%%%%%%%%%%%%%%%%
%%%%%%%%%%%%%%%%%%%%%%%%%%%%%%%%%%%%%%%%%%%%%%%%%%%%%%%%%%%%%%%%%%%%%
%%%%%%%%%%%%%%%%%%%%%%%%%%%%%%%%%%%%%%%%%%%%%
%%%%%%%%%%%%%%%%%%%%%%%%%%%%%%%%%%%%%%%%%%%%%
%%%%%%%%%%%%%%%%%%%%%%%%%%%%%%%%%%%%%%%%%%%%%%%%%%%%%%%%%%%%%%%%%%%%%
%%%%%%%%%%%%%%%%%%%%%%%%%%%%%%%%%%%%%%%%%%%%%%%%%%%%%%%%%%%%%%%%%%%%%
\begin{figure}[t!]
\begin{center}
\includegraphics[width=1.05\columnwidth]{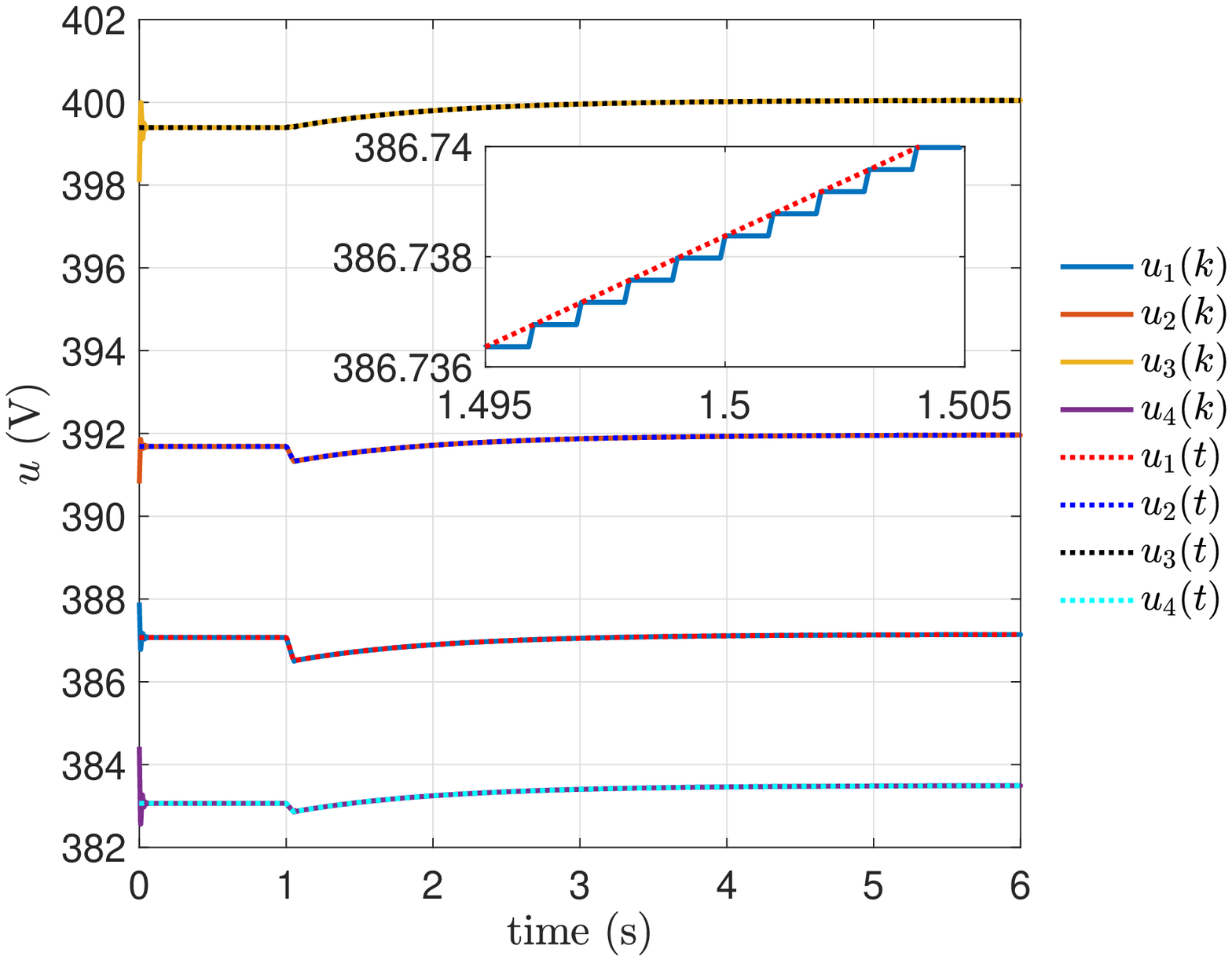}   
\caption{Continuous-time (dotted lines) and discrete-time (solid lines) control inputs.} 
\label{fig:out_con_u}
\end{center}
\end{figure}
%%%%%%%%%%%%%%%%%%%%%%%%%%%%%%%%%%%%%%%%%%%%%%%%%%%%%%%%%%%%%%%%%%%%%
%%%%%%%%%%%%%%%%%%%%%%%%%%%%%%%%%%%%%%%%%%%%%%%%%%%%%%%%%%%%%%%%%%%%%

\section{Proofs of Theorems in Section~\ref{sec:stab}}
\subsection{Proof of Theorem~\ref{thm:stab}}\label{app:stab}
Consider a Lyapunov candidate:
\begin{align*}
V(x_k,u_k) := S_K (x_k, u_k, \Delta_\delta x_k) + S_u (u_k).
\end{align*}
It follows from~\eqref{eq1:KP}  with $v_k = \Delta_\delta u_k$ and~\eqref{eq:con_KP} that
\begin{align*}
\Delta_\delta V(x_k, u_k) 
&= \Delta_\delta S_K(x_k, u_k, \Delta_\delta x_k) + \Delta_\delta S_u (u_k) \\
&\le  - W_K (\sigma_\delta x_k, u_k, \Delta_\delta \sigma_\delta x_k ) + \Delta_\delta u_k^\top z_k \\
&\qquad - c |\Delta_\delta u_k|^2 - \Delta_\delta u_k^\top z_k \\
&= - W_K (\sigma_\delta x_k, u_k, \Delta_\delta \sigma_\delta x_k ) - c |\Delta_\delta u_k|.
\end{align*}
From~\eqref{eq:sys_ex} and~\eqref{eq:con_stab}, $\hat z_k =\Delta_\delta u_k$.
Then, items (a) and (b) follow from a discrete-time version of the invariance principle, \emph{e.g.}, \cite[Theorem 13.3]{HC:11}.
\hfill\QED

%%%%%%%%%%%%%%%%%%%%%%%%%%%%%%%%%%%%%%%%%%%%%%%%%%%%%%%%%%%%%%%%%%%%%
%%%%%%%%%%%%%%%%%%%%%%%%%%%%%%%%%%%%%%%%%%%%%%%%%%%%%%%%%%%%%%%%%%%%%

\subsection{Proof of Proposition~\ref{prop:con_stab_IMM}}\label{app:con_stab_IMM}
Define
\begin{align}\label{pf1:con_stab_IMM}
S_u (u) := \frac{1}{2} |u -u^*|_{K_2}^2.
\end{align}
Then, it follows from~\eqref{eq2:IMM} and~\eqref{eq:con_stab_IMM} that
\begin{align}\label{pf2:con_stab_IMM}
\Delta_\delta S_u (u_k) 
&= \frac{|u_{k+1} -u^*|_{K_2}^2 - |u_k-u^*|_{K_2}^2}{2\delta} \nonumber\\
&= \frac{(u_{k+1} -u_k)^\top}{\delta} K_2 \frac{ u_k + u_{k+1} -2 u^*}{2} \nonumber\\
&= - (\Delta_\delta u_k)^\top K_2 (u^* - \sigma_\delta u_k )  \nonumber\\
&= - |\Delta_\delta u_k|_{K_1}^2 - (\Delta_\delta u_k)^\top z_k.
\end{align}
From $K_1 \succ 0$, we have~\eqref{eq:con_KP} for some $c > 0$.
\hfill\QED

%%%%%%%%%%%%%%%%%%%%%%%%%%%%%%%%%%%%%%%%%%%%%%%%%%%%%%%%%%%%%%%%%%%%%
%%%%%%%%%%%%%%%%%%%%%%%%%%%%%%%%%%%%%%%%%%%%%%%%%%%%%%%%%%%%%%%%%%%%%

\subsection{Proof of Theorem~\ref{lPHthm:stab}}\label{lPHapp:stab}
From~\eqref{eq:const} and~\eqref{lPHeq:sys}, $(x^*, u^*) \in \cE$ satisfies
\begin{align}\label{lPHpf1:stab}
(J - R) H x^* + B u^* =0
\end{align}
Note that $R \succ 0$ implies the invertibility of $J - R$. Also from $H \succ 0$, $(J - R) H$ is invertible. 
Thus, $x^*$ is uniquely determined if one specifies $u^*$. 
Moreover, $B$ is of full column rank.
Therefore, $u^*$ has a one-to-one correspondence with $(x^*, u^*) \in \cE$.

Next, the closed-loop system can be described as a system with the state $(x_k, x_{k+1}, u_k)$. This is well defined if $I_n/\delta  - (J - R) H/2$ and $A_s$ are invertible.

Now, we are ready to show exponential stability.
From~\eqref{IPHeq:storage} and~\eqref{pf1:con_stab_IMM}, we consider the following Lyapunov candidate:
\begin{align*}
V(x_k, u_k) = \frac{1}{2} (|\Delta_\delta x_k|_H^2 + |u_k -u^*|_{K_2}^2) \ge 0.
\end{align*}
From $H \succ 0$ and $K_2 \succ 0$, this becomes zero if and only if $\Delta_\delta x_k = 0$ and $u = u^*$.
Furthermore, under $H, R \succ 0$ and the invertibility of $I_n/\delta - (J - R) H/2$, one can conclude $x_k = x^*$ from~\eqref{lPHeq:sys}. Thus, $V(x,u)$ is positive definite at $(x^*, u^*)$.

It follows from~\eqref{lPHeq:KP} and~\eqref{pf2:con_stab_IMM} that
\begin{align*}
\Delta_\delta V(x_k, u_k) = - (|H \Delta_\delta \sigma_\delta x_k|_R^2 + |\Delta_\delta u_k|_{K_1}^2).
\end{align*}
From~$\Delta_\delta \sigma_\delta x_k = (x_{k+2}-  x_k)/2\delta$ with $H, R \succ 0$ and from $K_1 \succ 0$, we have
\begin{align}
\lim_{k \to \infty} \Delta_\delta \sigma_\delta x_k &= \lim_{k \to \infty} \frac{x_{k+2}-  x_k}{2\delta}=0
\label{lPHpf2:stab}\\
%%%
\lim_{k \to \infty} \Delta_\delta u_k &= \lim_{k \to \infty} \frac{u_{k+1}-  u_k}{\delta} =0.
\label{lPHpf3:stab}
\end{align}
Thus, the controller dynamics~\eqref{eq:con_stab_IMM} with $z_k = B^\top H \Delta_\delta \sigma_\delta x_k$ in~\eqref{lPHeq:sys_ex} lead to
\begin{align*}
 0 &= \lim_{k \to \infty} K_1 \Delta_\delta u_k \\
 &= \lim_{k \to \infty} ( K_2 (u^* - \sigma_\delta u_k) - B^\top H \Delta_\delta \sigma_\delta x_k )\\
 &= \lim_{k \to \infty} K_2 (u^* - \sigma_\delta u_k).
\end{align*}
From $K_2 \succ 0$, we have $\lim_{k \to \infty} \sigma_\delta u_k = u^*$.
Next, from~\eqref{lPHpf3:stab} and $\delta > 0$, there exists $\tilde u \in \bR^m$ such that $\lim_{k \to \infty} u_k = \tilde u$ and consequently $\lim_{k \to \infty} \sigma_\delta u_k = \tilde u$ from~\eqref{eq2:IMM}.
Therefore, we have $\tilde u = u^*$, \emph{i.e.}, $\lim_{k \to \infty} u_k = u^*$.

It remains to show the convergence of $x_k$ to $x^*$.
From $\delta > 0$,~\eqref{lPHpf2:stab} implies
\begin{align}\label{lPHpf4:stab}
\lim_{k \to \infty} x_{2k} = \tilde x', \quad
\lim_{k \to \infty} x_{2k+1} = \tilde x''
\end{align}
for some $\tilde x', \tilde x'' \in \bR^n$.
Accordingly,~\eqref{eq2:IMM} yields
\begin{align}\label{lPHpf5:stab}
\lim_{k \to \infty} \sigma_\delta x_k = \lim_{k \to \infty} \frac{x_k +x_{k+1}}{2} =\frac{\tilde x'+\tilde x''}{2}.
\end{align}
Thus, it follows from the system dynamics~\eqref{lPHeq:sys} that
\begin{align}\label{lPHpf6:stab}
\frac{\tilde x'' - \tilde x'}{\delta} &= \lim_{k \to \infty} \Delta_{\delta} x_{2k} \nonumber\\
&=  (J - R) H \frac{\tilde x' + \tilde x''}{2} + B u^* \nonumber\\
&=\lim_{k \to \infty} \Delta_{\delta} x_{2k+1} 
 = \frac{\tilde x' - \tilde x''}{\delta}.
\end{align}
This implies $\tilde x' = \tilde x''$.
Since a solution to~\eqref{lPHpf1:stab} is unique for given $u^*$, we have $\tilde x' = \tilde x'' = x^*$. 
Therefore, exponential stability of $(x^*, u^*)$ follows from a discrete-time version of the invariance principle.
\hfill\QED

%%%%%%%%%%%%%%%%%%%%%%%%%%%%%%%%%%%%%%%%%%%%%%%%%%%%%%%%%%%%%%%%%%%%%
%%%%%%%%%%%%%%%%%%%%%%%%%%%%%%%%%%%%%%%%%%%%%%%%%%%%%%%%%%%%%%%%%%%%%

%%%%%%%%%%%%%%%%%%%%%%%%%%%%%%%%%%%%%%%%%%%%%%%%%%%%%%%%%%%%%%%%%%%%%
%%%%%%%%%%%%%%%%%%%%%%%%%%%%%%%%%%%%%%%%%%%%%%%%%%%%%%%%%%%%%%%%%%%%%

\section{Proofs of Theorems in Section~\ref{sec:oc}}
\subsection{Proof of Theorem~\ref{thm:oc}}\label{app:oc}
We consider the following scalar valued function:
\begin{align*}
V (x_k, u_k, \Delta_\delta x_k)
:= S_K (x_k, u_k, \Delta_\delta x_k) + S_y(y_k).
\end{align*}
Then, it follows from~\eqref{eq1:KP} and~\eqref{eq1:con_oc_KP} that
\begin{align*}
\Delta_\delta V(x_k, u_k, \Delta_\delta  x_k) 
\le - W_K (\sigma_\delta x_k, u_k, \Delta_\delta \sigma_\delta x_k ).
\end{align*}
Since the closed-loop system is positively invariant on a compact set $\Omega$, a discrete-time version of the invariance principle concludes
\begin{align*}
\lim_{k \to \infty} W_K ( \sigma_\delta x_k, u_k,  \Delta_\delta \sigma_\delta x_k ) = 0
\end{align*}
on $\Omega$.
From~\eqref{eq2:KP}, this implies
\begin{align}\label{pf1:oc}
\lim_{k \to \infty} \Delta_\delta \sigma_\delta x_k = 0.
\end{align}
Combining~\eqref{eq:sys_d_fc} and~\eqref{eq2:con_oc_KP} leads to
\begin{align}\label{pf2:oc}
\lim_{k \to \infty} E^\top \sigma_\delta y_k = 0
\end{align}
on $\Omega$. Next,~\eqref{pf1:oc} implies
\begin{align*}
\lim_{k \to \infty} \Delta_\delta \sigma_\delta x_k = \lim_{k \to \infty} \frac{\sigma_\delta x_{k+1} - \sigma_\delta x_k}{\delta} = 0,
\end{align*}
and thus there exists $\tilde x \in \bR^n$ such that
\begin{align*}
\lim_{k \to \infty} \sigma_\delta x_k = \tilde x.
\end{align*}
Moreover, from $y_k = h_\delta (\sigma_\delta x_k, d)$ in~\eqref{eq:sys_d} and the continuity of $h_\delta$, we obtain
\begin{align}\label{pf3:oc}
\lim_{k \to \infty} y_k =  \tilde y : = h_\delta ( \tilde x, d).
\end{align}
Consequently, from~\eqref{eq:const} and the continuity of $\sigma_\delta$, we have
\begin{align}\label{pf4:oc}
\lim_{k \to \infty} \sigma_\delta y_k =  \tilde y.
\end{align}
Therefore,~\eqref{pf2:oc} --~\eqref{pf4:oc} lead to
\begin{align*}
\lim_{k \to \infty} E^\top y_k &=  E^\top \tilde y \\
&= \lim_{k \to \infty} E^\top \sigma_\delta y_k = 0.
\end{align*}
From the property of $E$, output consensus~\eqref{eq:oc} is achieved on $\Omega$.
\hfill\QED

%%%%%%%%%%%%%%%%%%%%%%%%%%%%%%%%%%%%%%%%%%%%%%%%%%%%%%%%%%%%%%%%%%%%%
%%%%%%%%%%%%%%%%%%%%%%%%%%%%%%%%%%%%%%%%%%%%%%%%%%%%%%%%%%%%%%%%%%%%%

\subsection{Proof of Theorem~\ref{lPHthm:oc}}\label{lPHapp:oc}
Using $S_K (\Delta x_k)$ in~\eqref{IPHeq:storage}, define
\begin{align*}
V(x_k, u_k, \rho_k) &:= S_K (\Delta x_k) + S_y (y_k, \rho_k)\\
 S_y (y_k, \rho_k) &:= ( |E^\top M y_k|^2 + |y_k - \rho_k|_K^2)/2.
\end{align*}
We first show that this is a positive definite function at some $(x^*, u^*, \rho^*)$.
From $S_y (y_k, \rho_k) = 0$ and $K \succ 0$, there exists $\alpha \in \bR$ such that $M \rho_k = M y_k = \alpha \1_m$.
Then,~\eqref{eq:const} and~\eqref{eq2:con_oc_IMM} imply $\Delta_\delta u_k = 0$, \emph{i.e.}, there exist $\tilde u \in \bR^m$ such that $u_k = \tilde u$.
Also, from $S_K (\Delta x_k)=0$, there exists $\tilde x \in \bR^n$ such that $x_k = \tilde x$.
Then, it follows from~\eqref{eq2:IMM} and~\eqref{lPHeq:sys_d} that
\begin{align}\label{lPHpf0:oc}
\left\{\begin{array}{r@{}l}
0 &{}= (J - R) H \tilde x + B \tilde u + d\\
M y_k &{} = \alpha \1_m.
\end{array}\right.
\end{align}
According to \cite[Proposition 1]{FKC:22} for continuous-time linear PHSs, if $R, H \succ 0$, $B$ is of full column rank, and $M$ is non-singular then given $d \in \bR^r$,~\eqref{lPHpf0:oc} has a unique solution $(x^*, u^*, \alpha)$. This implies that $V(x_k, u_k, \rho_k)$ is positive definite at $(x^*, u^*, \rho^*)$ with $\rho^*= \alpha M^{-1} \1_m$, and 
$(x^*, u^*, \rho^*)$ is a unique equilibrium of the closed-loop system.

Next, the closed-loop system can be described as a system with the state $(x_k, x_{k+1}, u_k, \rho_k)$. This is well defined if $I_n/\delta  - (J - R) H/2$ and $A_c$ are invertible.

Now, we are ready to show weighted output consensus. 
Similarly to Theorem~\ref{lPHthm:KP}, one can show that the system~\eqref{lPHeq:sys_d} is strictly Krasovskii passive with respect to a storage function $S_K (\Delta x_k)$ in~\eqref{IPHeq:storage},
where
\begin{align*}
W_K( \sigma_\delta x_k, u_k, \Delta_\delta \sigma_\delta x_k) = |H \Delta_\delta \sigma_\delta x_k|_R^2.
\end{align*}
Also, it follows from~\eqref{eq2:con_oc_IMM} that
\begin{align*}
&\Delta_\delta S_y(y_k, \rho_k) \\
&= (\Delta_\delta y_k)^\top M^\top  E E^\top M  \sigma_\delta y_k \\
&\qquad + (\sigma_\delta y_k - \sigma_\delta \rho_k )^\top K (\Delta_\delta y_k - \Delta_\delta \rho_k)\\
&= - (\Delta_\delta y_k)^\top \Delta_\delta  u_k
- (\Delta_\delta y_k)^\top K (\Delta_\delta y_k - \Delta_\delta \rho_k) \\
&\qquad + (\Delta_\delta \rho_k)^\top K (\Delta_\delta y_k - \Delta_\delta \rho_k)\\
&= - (\Delta_\delta y_k)^\top \Delta_\delta  u_k - |\Delta_\delta y_k -\Delta_\delta \rho_k|_K^2.
\end{align*}
Thus, we have
\begin{align*}
\Delta V (x_k, u_k, \rho_k) &= - |H \Delta_\delta \sigma_\delta x_k|_R^2 - |\Delta_\delta y_k -\Delta_\delta \rho_k|_K^2. 
\end{align*} 
Since $V (x_k, u_k, \rho_k)$ is positive definite and $H, R \succ 0$, this implies
\begin{subequations}
\begin{align}
&\lim_{k \to \infty} \Delta_\delta \sigma_\delta x_k = 0\label{lPHpf1:oc}\\
&\lim_{k \to \infty} K (\Delta_\delta y_k -\Delta_\delta \rho_k) =0
\label{lPHpf3:oc}
\end{align} 
\end{subequations}
for each $(x_0, u_0, \rho_0) \in \bR^n \times \bR^m \times \bR^m$.

From~\eqref{lPHpf1:oc}, we have~\eqref{lPHpf4:stab} and~\eqref{lPHpf5:stab} for some $\tilde x', \tilde x'' \in \bR^n$.
Since $B$ is of full column rank in~\eqref{lPHeq:sys_d}, $\{u_{2k}\}_{k \in \bZ_+}$ and $\{u_{2k+1}\}_{k \in \bZ_+}$ respectively converge to some $\tilde u' \in \bR^m$ and $\tilde u'' \in \bR^m$ that are
\begin{align*}
\tilde u' &:= B^+ \left( \frac{\tilde x'' - \tilde x'}{\delta} - (J - R) H \frac{\tilde x' + \tilde x''}{2} - d \right)\\
\tilde u'' &:= B^+ \left( \frac{\tilde x' - \tilde x''}{\delta} - (J - R) H \frac{\tilde x' + \tilde x''}{2} - d \right).
\end{align*}

On the other hand, it follows from~\eqref{lPHpf5:stab} and $y_k = B^\top H \sigma_\delta x_k$ that
\begin{align}\label{lPHpf4:oc}
\lim_{k \to \infty} y_k = \tilde y := B^\top H \frac{\tilde x' + \tilde x''}{2},
\end{align}
and consequently from~\eqref{eq2:con_oc_IMM} and~\eqref{lPHpf3:oc} that
\begin{align}\label{lPHpf5:oc}
\lim_{k\to \infty} \frac{u_{k+1} -u_k}{\delta} 
&= - \lim_{k\to \infty} M^\top E E^\top M \sigma_\delta y_k \nonumber\\
&= - M^\top E E^\top M \tilde y
\end{align}
This implies $\tilde u' - \tilde u'' = \tilde u'' - \tilde u'$, \emph{i.e.}, $\lim_{k \to \infty}\Delta_\delta u_k =0$. 
Therefore, from~\eqref{lPHpf4:oc} and~\eqref{lPHpf5:oc}, weighed output consensus~\eqref{eq:woc} is achieved.

It remains to show the uniqueness of $\alpha$. To this end, we show exponential stability of $(x^*, u^*, \rho^*)$. Since the solution $(x^*, u^*, \alpha)$ to~\eqref{lPHpf0:oc} is unique, exponential stability of $(x^*, u^*, \rho^*)$ implies uniquencess of the consensus value $\alpha$.

By a similar calculation as~\eqref{lPHpf6:stab}, one can show $\tilde x'=\tilde x''$, \emph{i.e.}, $\lim_{k \to \infty}\Delta_\delta x_k =0$, which further implies $\lim_{k \to \infty}\Delta_\delta y_k =0$. Again from~\eqref{lPHpf3:oc} with $K \succ 0$, we have $\lim_{k \to \infty}\Delta_\delta \rho_k =0$. Therefore, the unique equilibrium $(x^*, u^*, \rho^*)$ of the closed-loop system is exponentially stable by a discrete-time version of the invariance principle.  
\hfill\QED

%%%%%%%%%%%%%%%%%%%%%%%%%%%%%%%%%%%%%%%%%%%%%%%%%%%%%%%%%%%%%%%%%%%%%
%%%%%%%%%%%%%%%%%%%%%%%%%%%%%%%%%%%%%%%%%%%%%%%%%%%%%%%%%%%%%%%%%%%%%

%%%%%%%%%%%%%%%%%%%%%%%%%%%%%%%%%%%%%%%%%%%%%%%%%%%%%%%%%%%%%%%%%%%%%
%%%%%%%%%%%%%%%%%%%%%%%%%%%%%%%%%%%%%%%%%%%%%%%%%%%%%%%%%%%%%%%%%%%%%

%%%%%%%%%%%%%%%%%%%%%%%%%%%%%%%%%%%%%%%%%%%%%%%%%%%%%%%%%%%%%%%%%%%%%
%%%%%%%%%%%%%%%%%%%%%%%%%%%%%%%%%%%%%%%%%%%%%%%%%%%%%%%%%%%%%%%%%%%%%

%%%%%%%%%%%%%%%%%%%%%%%%%%%%%%%%%%%%%%%%%%%%%%%%%%%%%%%%%%%%%%%%%%%%%
%%%%%%%%%%%%%%%%%%%%%%%%%%%%%%%%%%%%%%%%%%%%%%%%%%%%%%%%%%%%%%%%%%%%%

\balance

\bibliographystyle{IEEEtran}        
\bibliography{ref}          

\end{document}